\documentclass{article}
\usepackage[a4paper, top=1.0in, bottom=1.0in, left=1.0in, right=1.0in]{geometry}
\usepackage[utf8]{inputenc}

\usepackage{amsmath,amsfonts}
\usepackage{array}
\usepackage[font=footnotesize,labelfont=bf]{caption}
\usepackage{subcaption}
\usepackage{textcomp}
\usepackage{stfloats}
\usepackage{url}
\usepackage{verbatim}
\usepackage{graphicx}
\usepackage[noadjust]{cite}

\usepackage{amsmath,amsfonts,amssymb,amsthm}
\usepackage{enumitem}
\usepackage{mathtools}
\usepackage{braket}
\usepackage{algorithm,algpseudocode}
\usepackage{siunitx}[=v2]
\usepackage{bm}
\usepackage{bbm}
\usepackage{mleftright}
\usepackage{booktabs}
\usepackage[usenames,dvipsnames]{xcolor}
\usepackage{pgfplots}
\usepackage{tikz}
\usepackage{color}
\usepackage{hyperref}
\hypersetup{colorlinks=true,
    linkcolor=BrickRed,
    citecolor=OliveGreen,
    urlcolor=MidnightBlue}
\usepackage{enumitem}
\usepackage{cleveref}

\usepackage{booktabs}
\usepackage{graphicx}
\usepackage{makecell}

\pgfplotsset{compat=1.9}
\usepgfplotslibrary{fillbetween}

\DeclareMathOperator*{\minimize}{min}
\DeclareMathOperator*{\maximize}{max}
\DeclareMathOperator*{\subjto}{subj. to}

\DeclareMathOperator{\tr}{tr}
\DeclareMathOperator{\epigraph}{epi}
\DeclareMathOperator{\hypograph}{hypo}
\DeclareMathOperator{\closure}{cl}
\DeclareMathOperator{\interior}{int}

\DeclareMathOperator{\domain}{dom}
\DeclareMathOperator{\diag}{diag}
\DeclareMathOperator{\image}{im}
\DeclareMathOperator{\vect}{vec}

\DeclarePairedDelimiterX{\divx}[2]{(}{)}{#1\mspace{1.5mu}\delimsize\|\mspace{1.5mu}#2}
\DeclarePairedDelimiterX{\divy}[2]{(}{)}{#1\mspace{1mu}\delimsize|\mspace{1mu}#2}
\DeclarePairedDelimiterX{\inp}[2]{\langle}{\rangle}{#1, #2}
\DeclarePairedDelimiterX{\norm}[1]{\lVert}{\rVert}{#1}
\DeclarePairedDelimiterX{\abs}[1]{\lvert}{\rvert}{#1}
\DeclarePairedDelimiterX{\bk}[2]{\langle}{\rangle}{#1 \delimsize\vert #2}

\newcommand*{\vertbar}{\rule[-1ex]{0.5pt}{2.5ex}}
\newcommand*{\horzbar}{\rule[.5ex]{2.5ex}{0.5pt}}

\makeatletter
\def\ALG@special@indent{%
    \ifdim\ALG@thistlm=0pt\relax
        \hskip-\leftmargin
    \else
        \hskip\ALG@thistlm
    \fi
}
\newcommand{\Input}[1]{\item[]\noindent\ALG@special@indent \textbf{Input:}\ #1}
\newcommand{\Output}[1]{\item[]\noindent\ALG@special@indent \textbf{Output:}\ #1}
\newcommand{\Indent}[1]{\item[]\noindent\ALG@special@indent \hspace{2.675em} #1}
\makeatother

\newcommand*\conj[1]{\overline{#1}}

\NewDocumentCommand{\grad}{e{_^}}{%
  \mathop{}\!
  \nabla
  \IfValueT{#1}{_{\!#1}}
  \IfValueT{#2}{^{#2}}
}

\makeatletter
\def\smallunderbrace#1{\mathop{\vtop{\m@th\ialign{##\crcr
   $\hfil\displaystyle{#1}\hfil$\crcr
   \noalign{\kern3\p@\nointerlineskip}%
   \tiny\upbracefill\crcr\noalign{\kern3\p@}}}}\limits}
\makeatother

\newtheorem{theorem}{Theorem}[section]
\newtheorem{lemma}[theorem]{Lemma}

\newtheorem{corollary}[theorem]{Corollary}
\newtheorem{definition}[theorem]{Definition}
\newtheorem{remark}[theorem]{Remark}
\newtheorem{example}[theorem]{Example}
\newtheorem{fact}[theorem]{Fact}
\newtheorem{assump}{Assumption}

\newcommand{\footremember}[2]{%
    \footnote{#2}
    \newcounter{#1}
    \setcounter{#1}{\value{footnote}}%
}
\newcommand{\footrecall}[1]{%
    \footnotemark[\value{#1}]%
}

\newcommand*{\diff}{}

\begin{document}

\title{Interior Point Methods for Structured Quantum Relative Entropy Optimization Problems}

\author{%
    Kerry He\footremember{monash}{Department of Electrical and Computer Systems Engineering, Monash University, Clayton VIC 3800, Australia. \url{{kerry.he1, james.saunderson}@monash.edu}} \and James Saunderson\footrecall{monash} \and Hamza Fawzi\footremember{cambridge}{Department of Applied Mathematics and Theoretical Physics, University of Cambridge, Cambridge CB3 0WA, United Kingdom. \url{h.fawzi@damtp.cam.ac.uk}}
}
\date{}

\maketitle

\begin{abstract}
    Quantum relative entropy \diff{optimization refers to a class of convex problems in which a linear functional is minimized} over an affine section of the epigraph of the quantum relative entropy function. Recently, the self-concordance of a natural barrier function was proved for this set\diff{, and various implementations of interior-point methods have been made available to solve this class of optimization problems.} In this paper, we show how common structures arising from applications in quantum information theory can be exploited to improve the efficiency of solving quantum relative entropy \diff{optimization problems} using interior-point methods. First, we show that the natural barrier function for the epigraph of the quantum relative entropy composed with positive linear operators is self-concordant, even when these linear operators map to singular matrices. Compared to modelling problems using the full quantum relative entropy cone, this allows us to remove redundant log\diff{-}determinant expressions from the barrier function and reduce the overall barrier parameter. Second, we show how certain slices of the quantum relative entropy cone exhibit useful properties which should be exploited whenever possible to perform certain key steps of interior-point methods more efficiently. We demonstrate how these methods can be applied to applications in quantum information theory, including quantifying quantum key rates, quantum rate-distortion functions, quantum channel capacities, and the ground state energy of Hamiltonians. Our numerical results show that these techniques improve computation times by up to several orders of magnitude, and allow previously intractable problems to be solved.
\end{abstract}

\section{Introduction} \label{sec:intro}

The (Umegaki) quantum relative entropy is an important function in quantum information theory used to measure the divergence between two quantum states. This is defined as
\begin{equation}\label{eqn:qre}
    S\divx{X}{Y} \coloneqq \tr[X (\log(X) - \log(Y))],
\end{equation}
on the domain $\{ (X, Y) \in \mathbb{H}^n_{+}\times\mathbb{H}^n_{+} : X \ll Y \}$, where we use $X \ll Y$ to mean $\ker(Y)\subseteq\ker(X)$, $\log$ denotes the matrix logarithm, and $\mathbb{H}^n_+$ denotes the set of positive semidefinite Hermitian matrices. It is known that the quantum relative entropy is jointly convex in both of its arguments~\cite{lieb1973proof, effros2009matrix}. Therefore minimizing the quantum relative entropy subject to affine constraints is a convex optimization problem. Many important quantities in quantum information theory involve solving these quantum relative entropy \diff{optimization problems}, i.e., conic optimization problems with respect to the quantum relative entropy cone
\begin{equation}
    \mathcal{K}_{\text{qre}} \coloneqq \{ (t, X, Y) \in \mathbb{R}\times\mathbb{H}^n_{+}\times\mathbb{H}^n_{+} : t \geq S\divx{X}{Y}, \ X \ll Y \}.
\end{equation}
Recently, self-concordance of the natural barrier function of the quantum relative entropy cone
\begin{equation}\label{eqn:qre-barrier}
    (t, X, Y) \mapsto -\log(t - S\divx{X}{Y}) - \log\det(X) - \log\det(Y),
\end{equation}
with domain $\mathbb{R}\times\mathbb{H}^n_{++}\times\mathbb{H}^n_{++}$ was established~\cite{fawzi2023optimal}. Together with recent advances in nonsymmetric conic programming~\cite{skajaa2015homogeneous,myklebust2014interior,myklebust2015primal,tunccel2001generalization,nesterov2012towards,nemirovski2005cone}, \diff{there has been a surge of interest in solving} quantum relative entropy \diff{optimization problems} using primal-dual interior-point algorithms for nonsymmetric cones~\cite{coey2023performance,karimi2024domain,karimi2023efficient}. \diff{However, a limitation of interior-point methods is that they do not generally scale well to large problem dimensions. Current implementations of these methods are limited to solving small to moderately sized quantum relative entropy optimization problems. In this work, we focus on improving the computational efficiency of interior-point methods applied to solving quantum relative entropy optimization problems.} In particular, we study how we can exploit common structures arising in problems in quantum information theory to efficiently implement certain key steps of these interior-point algorithms.

\paragraph{Quantum relative entropy \diff{optimization}}
Several techniques have been proposed for solving quantum relative entropy \diff{optimization problems}. One line of work is to use first-order methods, e.g.,~\cite{zinchenko2010numerical,winick2018reliable,you2022minimizing,he2023bregman,he2023efficient}. However, these methods typically require tailored algorithms and proofs of convergence for each different problem instance, and currently there is no general purpose first-order algorithm which can solve general quantum relative entropy \diff{optimization problems}. Additionally, these techniques usually \diff{exhibit sublinear convergence rates (i.e., $O(1/\varepsilon)$), and} do not converge quickly to high accuracy solutions.

Alternatively, the quantum relative entropy cone can be approximated to arbitrarily small precision by using linear matrix inequalities~\cite{fawzi2019semidefinite}. This approximation allows us to use relatively mature techniques and software for semidefinite programming to solve these problems. However, \diff{this technique uses (multiple) semidefinite constraints with $2n^2\times 2n^2$ matrices to approximate a quantum relative entropy involving $n\times n$ matrices. Therefore, this technique scales poorly to large problem dimensions, and is only practical for relatively small problems.}

Another recent line of work aims to solve quantum relative entropy \diff{optimization problems} using interior-point algorithms. \diff{These methods are attractive as although each step of an interior-point method is more expensive to perform compared to first-order methods, these algorithms typically require many fewer iterations to converge to high accuracy solutions. Additionally, interior-point methods can be elegantly extended to account for any additional combination of conic constraints, which can be less straightforward to do for first-order methods.} 

Tailored interior-point algorithms have been proposed in~\cite{faybusovich2020self,faybusovich2022long,hu2022robust} to compute the relative entropy of entanglement and quantum key rates. More generally, we can use nonsymmetric cone programming techniques to optimize over the quantum relative entropy cone. In~\cite{tunccel2001generalization,myklebust2014interior,myklebust2015primal,nesterov2012towards,nemirovski2005cone}, it was shown how primal-dual interior-point algorithms for symmetric cone programming (i.e., linear, semidefinite, and second-order cone programming), which are known to work well in practice, could be extended to solve cone programs involving more general nonsymmetric cones while retaining the same worst case iteration complexity bounds $O(\sqrt{\nu}\log(1/\varepsilon))$, where $\nu$ is the barrier parameter associated with the cone. Practical implementations of nonsymmetric cone programming which use these ideas have recently been released~\cite{skajaa2015homogeneous,coey2023performance,karimi2024domain,papp2022alfonso,dahl2022primal}. Notably, implementations of the quantum relative entropy cone are available in Hypatia~\cite{coey2023performance} and DDS~\cite{karimi2024domain,karimi2023efficient}.

\paragraph{Exploiting structure in conic progams}
Consider the standard form conic program
\begin{equation}\label{eqn:standard-form-conic}
    \minimize_{x\in\mathbb{R}^n} \quad \inp{c}{x}, \quad \subjto \quad Ax = b, \ x \in\mathcal{K},
\end{equation}
where $c\in\mathbb{R}^n$, $A\in\mathbb{R}^{p\times n}$, $b\in\mathbb{R}^p$, and $\mathcal{K}\in\mathbb{R}^n$ is a proper (i.e., closed, pointed, and full-dimensional) convex cone, and let $F:\interior\mathcal{K}\rightarrow\mathbb{R}$ be a logarithmically-homogeneous self-concordant barrier (defined in Section~\ref{subsec:concordant-prelim}) for $\mathcal{K}$. The main bottleneck in primal-dual interior-point methods is typically in solving a linear system of equations, known as the Newton equations, involving the Hessian of the barrier function $\grad^2 F(x)$ and the constraint matrix $A$. One common way of solving the Newton equations is to use a suitable block elimination ordering (see Appendix~\ref{appdx:pdipm}) to reduce the problem to solving a linear system with the Schur complement matrix $A \grad^2 F(x)^{-1} A^\top$.

For linear and semidefinite programming, i.e., $\mathcal{K}=\mathbb{R}^n_+$ and $\mathcal{K}=\mathbb{H}^n_+$, respectively, the Schur complement matrix is relatively easy to construct as there exist simple expressions for the inverse Hessian map $\grad^2 F(x)^{-1}$ of the logarithmic and log determinant barriers used by each of these cones. Additionally, it is well known that certain structures in $A$, such as sparsity, can be exploited to make constructing the Schur complement matrix more efficient~\cite{andersen2011interior,fujisawa1997exploiting,benson2008dsdp5,vandenberghe2015chordal}. 
However, for nonsymmetric programming in general, it is not always clear if efficient expressions exist for the inverse Hessian map of the barrier function. Notably, existing implementations of interior-point algorithms for quantum relative entropy \diff{optimization problems}, i.e., $\mathcal{K}=\mathcal{K}_{\textnormal{qre}}$, construct and factor the full Hessian matrix of the barrier function~\eqref{eqn:qre-barrier}, which ends up being the dominant cost of these algorithms. In~\cite{karimi2023efficient}, a heuristic method which approximates the Hessian of quantum relative entropy by discarding off-diagonal blocks of the matrix is proposed, which simplifies the process of constructing and factoring the Hessian. However, we are not aware of any analysis of when this approximation works well.


Alternatively, sometimes the Hessian of the barrier function~\eqref{eqn:qre-barrier} can be simplified by exploiting identities which quantum relative entropy satisfies along certain slices of the quantum relative entropy cone. We illustrate \diff{two} examples of this idea which we will develop later in the paper.


\begin{example}\label{exmp:structrure}
    Consider the positive semidefinite $2\times2$ block matrix $X\in\mathbb{H}^{2n}_{+}$ whose $i,j$-th block is given by $X_{ij}\in\mathbb{H}^n$. One can show that the function
    \begin{equation}\label{eqn:dummy-quant-cond}
        X\mapsto S \biggl(
        \begin{bmatrix}
            X_{11} & X_{12} \\ X_{12}^\dag & X_{22}
        \end{bmatrix} \bigg\|
        \begin{bmatrix}
            X_{11}+X_{22} & 0\\ 0& X_{11}+X_{22}
        \end{bmatrix} \biggr),
    \end{equation}
    is equivalent to
    \begin{equation*}
        X \mapsto \tr[X\log(X)] - \tr[(X_{11}+X_{22})\log(X_{11}+X_{22})],
    \end{equation*}
    for all $X\in\mathbb{H}^{2n}_{+}$. We can represent the Hessian matrix of the second function in the form $H_1 - A^\dag H_2A$, where $H_1$ is the Hessian matrix of $X\in\mathbb{H}^{2n}_{+}\mapsto\tr[X\log(X)]$, $H_2$ is the Hessian matrix of $Y\in\mathbb{H}^{n}_{+}\mapsto\tr[Y\log(Y)]$, and $A$ represents the linear map $X\mapsto X_{11}+X_{22}$. By interpreting this Hessian as a low-rank perturbation of $H_1$, and recognizing that $H_1$ and $H_2$ are both easily invertible (see Remark~\ref{rem:easy-inverse-hessian}), we can employ the matrix inversion lemma (see Appendix~\ref{appdx:matrix-inversion-lemma}) to efficiently solve linear systems with $H_1 - A^\dag H_2A$. 

    The function in~\eqref{eqn:dummy-quant-cond} is related to quantum conditional entropy functions. We provide a more detailed discussion of these function in Example~\ref{exmp:conditional}, and provide a proof of a generalization of the identity~\eqref{eqn:dummy-quant-cond} in Lemma~\ref{lem:partial-trace-entropy}.
\end{example}

\begin{example}
    \diff{An important application of quantum relative entropy optimization is to compute quantum key rates, which requires solving a convex optimization problem of the form}
    \begin{equation*}
        \min_{X\in\mathbb{H}^n} \quad S\divx{\mathcal{G}(X)}{\mathcal{Z}(\mathcal{G}(X))} \quad \subjto \quad \mathcal{A}(X) = b, \ X\succeq0,
    \end{equation*}
    \diff{where $\mathcal{G}:\mathbb{H}^n\rightarrow\mathbb{H}^{mr}$ and $\mathcal{Z}:\mathbb{H}^{mr}\rightarrow\mathbb{H}^{mr}$ are positive linear maps, $\mathcal{A}:\mathbb{H}^n\rightarrow\mathbb{R}^p$ is a linear map, and $b\in\mathbb{R}^p$. We introduce this example in more detail in Section~\ref{subsec:qkd-experiments}.}
    
    \diff{The linear maps $\mathcal{G}$ and $\mathcal{Z}$ typically possess properties which can be exploited to make computing the Hessian of the objective function easier. Notably, the linear map $\mathcal{G}$ often maps to low rank matrices, and the linear map $\mathcal{Z}$ maps all off-diagonal blocks of a matrix to zero for a given block structure. By using these properties, it was shown in~\cite{hu2022robust} how facial reduction and a similar decomposition as shown in Example~\ref{exmp:structrure} could simplify the computation of the Hessian matrix. Later in Section~\ref{subsec:qkd-experiments}, we show how the block-diagonal structure of the image of $\mathcal{Z}$, together with further knowledge about the structure of $\mathcal{G}$, can be further exploited to simplify this computation.}

\end{example}


\paragraph{Contributions}
In our work, we focus on quantum relative entropy \diff{optimization problems} of the form, or closely related to the form,
\begin{equation}\label{eqn:qrep-base}
    \min_{X\in\mathbb{H}^n} \quad S\divx{\mathcal{G}(X)}{\mathcal{H}(X)} \quad \subjto \quad \mathcal{A}(X) = b, \ X\succeq0,
\end{equation}
where $\mathcal{G}:\mathbb{H}^n\rightarrow\mathbb{H}^m$ and $\mathcal{H}:\mathbb{H}^n\rightarrow\mathbb{H}^m$ are positive linear maps (i.e., linear operators which map positive semidefinite matrices to positive semidefinite matrices) satisfying $\mathcal{G}(X) \ll \mathcal{H}(X)$ for all $X\in\mathbb{H}^n_{++}$, and $\mathcal{A}:\mathbb{H}^n\rightarrow\mathbb{R}^p$ and $b\in\mathbb{R}^p$ encode the linear constraints. Note that $\mathcal{G}$ and/or $\mathcal{H}$ may be defined to map nonsingular matrices to singular matrices. To model this problem as a standard form conic program~\eqref{eqn:standard-form-conic}, an obvious choice is to let $\mathcal{K}=\mathcal{K}_{\textnormal{qre}}\times\mathbb{H}^n_+$, in which case we are interested in the following slice of the barrier of this cone
\begin{align}\label{eqn:barrier-qre-psd}
    (t, X) \mapsto &-\log(t - S\divx{\mathcal{G}(X)}{\mathcal{H}(X)}) - \log\det(\mathcal{G}(X)) - \log\det(\mathcal{H}(X)) - \log\det(X),
\end{align}
defined on $\mathbb{R}\times\mathbb{H}^n_{++}$, which has a total barrier parameter of $2m+n+1$. 

Our contributions are twofold. First, instead of using the aforementioned modelling strategy using the full quantum relative entropy cone, we propose directly model\diff{ling} the quantum relative entropy \diff{optimization problem}~\eqref{eqn:qrep-base} using the following slice of the quantum relative entropy cone
\begin{equation}\label{eqn:main-set}
    \mathcal{K}_{\textnormal{qre}}^{\mathcal{G}, \mathcal{H}} \coloneqq \closure \{ (t, X)\in\mathbb{R}\times\mathbb{H}^n_{++} : t \geq S\divx{\mathcal{G}(X)}{\mathcal{H}(X)} \}.
\end{equation}
We show in Section~\ref{sec:opt-barrier} that a natural barrier function for this slice of the quantum relative entropy cone is self-concordant. \diff{Additionally, this barrier is optimal in the sense it has the smallest barrier parameter out of any self-concordant barrier for the epigraph~\eqref{eqn:main-set}. Note, however, that this does not preclude the possibility of modeling the optimization problem using a different cone and barrier with a smaller barrier parameter.} We summarize this result in the following corollary, which is a special case of a more general result we present later in Theorem~\ref{thm:main}.
\begin{corollary}\label{cor:main}
    Let $\mathcal{G}:\mathbb{H}^n\rightarrow\mathbb{H}^m$ and $\mathcal{H}:\mathbb{H}^n\rightarrow\mathbb{H}^m$ be positive linear maps satisfying $\mathcal{G}(X) \ll \mathcal{H}(X)$ for all $X\in\mathbb{H}^n_{++}$. Then 
    \begin{equation}\label{eqn:main-barrier}
        (t, X) \mapsto -\log(t - S\divx{\mathcal{G}(X)}{\mathcal{H}(X)}) - \log\det(X),
    \end{equation}
    defined on $\mathbb{R}\times\mathbb{H}^n_{++}$ is an $(n+1)$-logarithmically homogeneous self-concordant barrier for $\mathcal{K}_{\textnormal{qre}}^{\mathcal{G}, \mathcal{H}}$. Moreover, this barrier is optimal in the sense that any self-concordant barrier for $\mathcal{K}_{\textnormal{qre}}^{\mathcal{G}, \mathcal{H}}$ has parameter at least $n+1$.
\end{corollary}

There are two main advantages of working with $\mathcal{K}_{\textnormal{qre}}^{\mathcal{G}, \mathcal{H}}$ and its barrier compared to working with the full quantum relative entropy cone. First, by using a priori knowledge that $\mathcal{G}$ and $\mathcal{H}$ are positive linear maps, we do not require redundant positive semidefinite constraints which introduce additional log determinant expressions to the barrier function and increase the total barrier parameter by $2m$ (compare~\eqref{eqn:barrier-qre-psd} and~\eqref{eqn:main-barrier}). Second, the barrier~\eqref{eqn:main-barrier} is well-defined even when $\mathcal{G}$ and/or $\mathcal{H}$ map to singular matrices, whereas~\eqref{eqn:barrier-qre-psd} is not well-defined under the same setting due to the redundant log determinant terms.

Second, in Section~\ref{sec:structure}, we show how \diff{certain structures in $\mathcal{G}$ and $\mathcal{H}$, which commonly arise in quantum information theory, \diff{can be exploited to simplify} solving the Newton equations. We achieve this by extending the basic idea discussed in Example~\ref{exmp:structrure}.}
Structures we study include block diagonal, low-rank, and quantum conditional entropy-like structures. 
Using these techniques, we implement custom cone oracles which can be used with the generic conic primal-dual interior point software Hypatia~\cite{coey2023performance}, which are available at
\begin{center}
    \url{https://github.com/kerry-he/qrep-structure}.
\end{center}
\diff{In Section~\ref{sec:exp},} we show how our techniques can be applied to various applications from quantum information theory, including computing the quantum key rate, the quantum rate-distortion function, quantum channel capacities, and the ground state energy of Hamiltonians. Our numerical results demonstrate that the methods presented in this paper can improve computation times and memory requirements by up to several orders of magnitude, allowing us to solve problems with dimensions that were previously out of reach. 

At the time of writing this paper, we were made aware of~\cite{lorente2024quantum} which independently proposed using interior-point methods to optimize over $\mathcal{K}_{\textnormal{qre}}^{\mathcal{G}, \mathcal{H}}$ with a focus on computing quantum key rates. While~\cite{lorente2024quantum} applies this technique to more quantum key rate examples, and benchmarks against state-of-the-art algorithms for computing quantum key rates, our work analyses the self-concordance properties of the barrier~\eqref{eqn:main-barrier}, considers a broader range of structural properties in $\mathcal{G}$ and $\mathcal{H}$, and a wider range of applications beyond quantum key distribution. 

\diff{We have also recently released a software package, QICS, containing high-performance Python implementations of the ideas explored in this paper. See~\cite{he2024qics} for detailed documentation and benchmarks for QICS.}

\section{Preliminaries} \label{sec:preliminaries}

\paragraph{Notation}
We use $\mathbb{R}^n$, $\mathbb{C}^n$, and $\mathbb{H}^n$ to denote the set of $n$-dimensional real vectors, $n$-dimensional complex vectors, and $n\times n$ Hermitian matrices, respectively. Additionally, we use $\mathbb{R}^n_+$ and $\mathbb{R}^n_{++}$ to denote the nonnegative orthant and its interior, respectively, and use $\mathbb{H}^n_+$ and $\mathbb{H}^n_{++}$ to denote the positive semidefinite cone and its interior, respectively. We will also use the notation $X\succeq0$ and $X\succ0$ to mean $X\in\mathbb{H}^n_{+}$ and $X\in\mathbb{H}^n_{++}$, respectively. For $x,y\in\mathbb{R}^n$ and $X,Y\in\mathbb{H}^n$, we define the standard inner products $\inp{x}{y}=x^\top y$ and $\inp{X}{Y}=\tr[XY]$, respectively. For a complex matrix $X\in\mathbb{C}^{n\times m}$, we use $\conj{X}$ to denote its elementwise conjugate, and $X^\dag$ to denote its conjugate transpose. We use $\mathbb{I}_n$ to denote the $n\times n$ identity matrix, and drop the subscript if the dimension of the matrix is clear from the context. We use $\{ e_i \}_{i=1}^n$ to denote the standard basis for either $\mathbb{R}^n$ or $\mathbb{C}^n$. We use $\closure$, $\epigraph$, $\hypograph$, $\ker$, and $\image$ to denote the closure of a set, the epigraph or hypograph of a function, and the kernel or image of a linear operator, respectively. A matrix   $U\in\mathbb{C}^{n\times n}$ is unitary if $U^\dag U = UU^\dag = \mathbb{I}$. We say that a matrix $V\in\mathbb{C}^{m\times n}$ is an \emph{isometry matrix} if $m\geq n$ and $V^\dag V = \mathbb{I}$. 


\paragraph{Derivatives}
For a finite dimensional real vector space $\mathbb{V}$, consider a $C^3$ function $f:\domain f\rightarrow\mathbb{R}$ with open domain $\domain f \subset\mathbb{V}$. We denote the first, second, and third directional derivative of $f$ at $x\in\domain f$ in the direction $h\in\mathbb{V}$ as $\grad f(x)[h]$, $\grad^2 f(x)[h, h]$, and $\grad^3 f(x)[h, h, h]$, respectively. We will also use the notation $\grad f(x)\in\mathbb{V}$ and $\grad^2 f(x)[h]\in\mathbb{V}$ to denote the gradient vectors of $f$ and $\grad f(x)[h]$, respectively, and use $\grad^2 f(x)$ to denote the real Hessian matrix of $f$. We provide additional details regarding the derivatives of spectral functions and barrier functions in Appendix~\ref{appdx:derivatives}.


\paragraph{Product spaces}
For two matrices $X\in\mathbb{C}^{n\times m}$ and $Y\in\mathbb{C}^{p\times q}$, we define the Kronecker product $X\otimes Y\in\mathbb{C}^{np\times mq}$ as the block matrix
\begin{equation*}
    X\otimes Y = \begin{bmatrix}
        X_{11}Y & \ldots & X_{1m}Y \\
        \vdots & \ddots & \vdots \\
        X_{n1}Y & \ldots & X_{nm}Y
    \end{bmatrix},
\end{equation*}
where $X_{ij}$ denotes the $i,j$-th element of $X$. A related class of operators are the partial traces $\tr_1^{n,m}:\mathbb{C}^{nm\times nm}\rightarrow\mathbb{C}^{m\times m}$ and $\tr_2^{n,m}:\mathbb{C}^{nm\times nm}\rightarrow\mathbb{C}^{n\times n}$, which are defined as the adjoints of the Kronecker product with the identity matrix, i.e., $(\tr_1^{n,m})^\dag(Y)=\mathbb{I}_n\otimes Y$ for all $Y\in\mathbb{C}^{m\times m}$, and $(\tr_2^{n,m})^\dag(X)=X\otimes\mathbb{I}_m$ for all $X\in\mathbb{C}^{n\times n}$. Another interpretation of the partial traces is that they are the unique linear operators which satisfy
\begin{align*}
    \tr_1^{n,m}(X\otimes Y) &= \tr[X] Y\\
    \tr_2^{n,m}(X\otimes Y) &= \tr[Y] X,
\end{align*}
for all $X\in\mathbb{C}^{n\times n}$ and $Y\in\mathbb{C}^{m\times m}$. We will also use $\oplus$ to denote the direct sum of matrices, i.e., for matrices $X\in\mathbb{C}^{n\times m}$ and $Y\in\mathbb{C}^{p\times q}$, we define $X\oplus Y\in\mathbb{C}^{(n+p)\times(m+q)}$ as
\begin{equation*}
    X\oplus Y = \begin{bmatrix}
        X & 0 \\ 0 & Y
    \end{bmatrix}.
\end{equation*}


\paragraph{Entropy}
\sloppy Earlier, we introduced the quantum relative entropy function~\eqref{eqn:qre}. Here, we will briefly introduce some related functions to which we refer throughout the paper. A closely related function is the \emph{(von Neumann) quantum entropy}, which, through some overloading of notation, is defined as
\begin{equation}
    S(X) \coloneqq -\tr[X\log(X)],
\end{equation}
over the domain $\mathbb{H}^n_+$, and is a concave function. We define the \emph{quantum entropy cone} as the hypograph of the homogenized quantum entropy function
\begin{equation}
    \mathcal{K}_{\textnormal{qe}} \coloneqq \closure \{ (t, X, y) \in \mathbb{R}\times\mathbb{H}^n_{++}\times\mathbb{R}_{++} : t \geq -yS(X/y) \},
\end{equation}
and note that self-concordance of the natural barrier for the quantum entropy cone follows from the same result for the quantum relative entropy cone (see also~\cite{coey2023conic}). We also define a function closely related to the quantum entropy,
\begin{equation}\label{eqn:funny-entropy}
    S_C(X) \coloneqq -\tr[C\log(X)],
\end{equation}
over the domain $\{X\in\mathbb{H}^n_+ : X \ll C\}$, for some matrix $C\in\mathbb{H}^n$. In this notation, we can represent quantum relative entropy as $S\divx{X}{Y} = -S(X) + S_X(Y)$. The classical counterparts to these quantum entropies are the \emph{(Shannon) classical entropy},
\begin{equation}
    H(x) \coloneqq -\sum_{i=1}^n x_i\log(x_i),
\end{equation}
defined over $\mathbb{R}^n_{+}$, and the \emph{classical relative entropy} (or the \emph{Kullback–Leibler divergence}),
\begin{equation}\label{eqn:cre}
    H\divx{x}{y} \coloneqq \sum_{i=1}^n x_i\log(x_i / y_i),
\end{equation}
defined over $\{(x, y) \in\mathbb{R}^n_+\times\mathbb{R}^n_+ : \diag(x) \ll \diag(y) \}$. The \emph{classical relative entropy cone} is
\begin{equation}
    \mathcal{K}_{\textnormal{cre}} \coloneqq \{ (t, x, y) \in \mathbb{R}\times\mathbb{R}^n_{+}\times\mathbb{R}^n_+ : t \geq H\divx{x}{y}, \ \diag(x) \ll \diag(y) \}.
\end{equation}
Self-concordance of the natural barrier for this cone follows from the same result for the quantum relative entropy cone (see also~\cite[Theorem 1]{karimi2024domain}). These classical counterparts will be useful as quantum entropy and quantum relative entropy reduce to classical entropy and classical relative entropy, respectively, when the matrix arguments are diagonal.



\section{Optimal self-concordant barriers}\label{sec:opt-barrier}
In this section, we establish (optimal) self-concordance of the natural barrier of the epigraph of quantum relative entropy composed with positive linear maps, as well as for the epigraphs of related functions. To present our main result, we first introduce the following class of divergence measures called quasi-entropies, which were studied in~\cite{hiai2017different,petz1986quasi} and can be interpreted as a quantum extension of the classical $f$-divergences. For a function $g:(0,\infty)\rightarrow\mathbb{R}$, the quasi-entropy $S_g:\mathbb{H}^n_{++}\times\mathbb{H}^n_{++}\rightarrow\mathbb{R}$ is defined as
\begin{equation}
    S_g\divx{X}{Y} = \Psi[P_g(X\otimes\mathbb{I}, \mathbb{I}\otimes \conj{Y})],
\end{equation}
where $\Psi:\mathbb{H}^{n^2}\rightarrow\mathbb{R}$ is the unique linear operator satisfying $\Psi[X\otimes \conj{Y}]=\tr[XY]$, and
\begin{equation}
    P_g(X, Y) = X^{1/2} g(X^{-1/2} Y X^{-1/2}) X^{1/2},
\end{equation}
represents the noncommutative perspective of $g$. When $g$ is operator concave, meaning that for all integers $n$,  $X,Y\in\mathbb{H}^n_{++}$, and $\lambda\in[0, 1]$, we have
\begin{equation}
    g(\lambda X + (1-\lambda)Y) \succeq \lambda g(X) + (1-\lambda)Y,
\end{equation} 
then the quasi-entropy is a concave function~\cite[Proposition 3.10]{hiai2017different}. Notably, when $g(x)=-\log(x)$, it turns out that $S_g$ is the quantum relative entropy function. Under certain circumstances (see~\cite[Theorem B.3]{fawzi2023optimal}), is possible to extend the quasi-entropy function to the boundary of the positive semidefinite cone by considering the appropriate limit
\begin{equation}\label{eqn:quasi-entropy-limit}
    S_g\divx{X}{Y} = \lim_{\varepsilon\downarrow0} \Psi[ P_g( (X + \varepsilon\mathbb{I}) \otimes \mathbb{I}, \mathbb{I} \otimes (\conj{Y} + \varepsilon\mathbb{I})  ].
\end{equation}
When $X,Y\in\mathbb{H}^n_{++}$, this limit coincides with the original definition of quasi-entropy. For example, when $g(x)=-\log(x)$, this limit is well defined for all $X,Y\in\mathbb{H}^n_+$ such that $X \ll Y$. In~\cite{fawzi2023optimal}, it was shown that the barrier function
\begin{equation}
    (t, X, Y) \mapsto -\log(S_g\divx{X}{Y} - t) - \log\det(X) - \log\det(Y),
\end{equation}
defined on $\mathbb{R}\times\mathbb{H}^n_{++}\times\mathbb{H}^n_{++}$ is an optimal $(2n+1)$-logarithmically homogeneous self-concordant barrier for $\closure\hypograph S_g$ when $g$ is operator concave. However, in many problems arising in quantum information theory, we are interested in minimizing quasi-entropies composed with positive linear maps, i.e., functions of the form
\begin{equation}\label{eqn:quasi-entropy-N}
    S_g^{\mathcal{N}_1, \mathcal{N}_2}(X) = S_g\divx{\mathcal{N}_1(X)}{\mathcal{N}_2(X)},
\end{equation}
defined over $\mathbb{H}^n_{++}$, for positive linear operators $\mathcal{N}_1:\mathbb{H}^n\rightarrow\mathbb{H}^m$ and $\mathcal{N}_2:\mathbb{H}^n\rightarrow\mathbb{H}^m$. Often, these linear operators map to singular matrices, in which some care is needed to ensure the function is well-defined. The following set of assumptions outlines all possible scenarios when this function is well-defined when using the limit~\eqref{eqn:quasi-entropy-limit}, and closely mirrors the results from~\cite[Theorem B.3]{fawzi2023optimal}.
\begin{assump}\label{assump:well-defined-N}
    Consider a function $g:(0, \infty)\rightarrow\mathbb{R}$ and positive linear maps $\mathcal{N}_1:\mathbb{H}^n\rightarrow\mathbb{H}^m$ and $\mathcal{N}_2:\mathbb{H}^n\rightarrow\mathbb{H}^m$. Let $\hat{g}(x)\coloneqq xg(1/x)$ denote the transpose of $g$, and let $g(0^+)\coloneqq\lim_{x\rightarrow0}g(x)$ and $\hat{g}(0^+)\coloneqq\lim_{x\rightarrow0}\hat{g}(x)$. Assume that any one of the following conditions is satisfied
    \begin{enumerate}[label=(\roman*), ref=\ref{assump:well-defined-N}(\roman*),leftmargin=25pt]
        \item $\mathcal{N}_1(X) \ll \mathcal{N}_2(X)$ for all $X\in\mathbb{H}^n_{++}$ and $\hat{g}(0^+)>\infty$, or \label{assump:well-defined-N-i}
        \item $\mathcal{N}_2(X) \ll \mathcal{N}_1(X)$ for all $X\in\mathbb{H}^n_{++}$ and $g(0^+)>\infty$, or \label{assump:well-defined-N-ii}
        \item $\ker \mathcal{N}_1(X) = \ker \mathcal{N}_2(X)$ for all $X\in\mathbb{H}^n_{++}$, or \label{assump:well-defined-N-iii}
        \item $\hat{g}(0^+)>\infty$ and $g(0^+)>\infty$. \label{assump:well-defined-N-iv}
    \end{enumerate}
\end{assump}

Our main theorem extends the result in~\cite{fawzi2023optimal} by giving logarithmically homogeneous self-concordant barriers for the epigraph of quasi-entropy when composed with positive linear maps, which possibly map to singular matrices.

\diff{Before stating our main theorem, we first recall the definition of logarithmically homogeneous self-concordant barriers.
Let $\mathbb{V}$ be a finite-dimensional real vector space, and let $F:\domain F \subset \mathbb{V} \rightarrow \mathbb{R}$ be a closed, convex, $C^3$ function with open domain $\domain F \subset \mathbb{V}$. We say that $F$ is self-concordant if 
\begin{equation}
    \abs{ \grad^3 F(x)[h, h, h] } \leq 2 ( \grad^2 F(x)[h, h] )^{3/2},
\end{equation}
for all $x\in\domain F$ and $h\in\mathbb{V}$. If, additionally, we have
\begin{equation}
    2\grad F(x)[h] - \grad^2 F(x)[h, h] \leq \nu,
\end{equation}
for all $x\in\domain F$ and $h\in\mathbb{V}$, then we say \diff{that} $F$ is a $\nu$-self-concordant barrier for the set $\closure\domain F$. \diff{Separately, if $\domain F$ is a convex cone}, then we say that $F$ is $\nu$-logarithmically homogeneous if there is a constant $\nu\geq1$ such that
\begin{equation}\label{eqn:logarithmic-homogeneity}
    F(\tau x) = F(x) - \nu \log(\tau),
\end{equation}
for all $x\in\interior\domain F$ and $\tau>0$. Note that if $F$ is \diff{both self-concordant and $\nu$-logarithmically homogeneous}, then it must also be a $\nu$-self-concordant barrier~\cite[Lemma 5.4.3]{nesterov2018lectures}.}

\diff{We are now ready to state our main theorem.}
\begin{theorem}\label{thm:main}
    Let $g:(0, \infty)\rightarrow\mathbb{R}$ be an operator concave function, and let $P_g$ be its noncommutative perspective. Let $\mathcal{N}_1:\mathbb{H}^n\rightarrow\mathbb{H}^m$ and $\mathcal{N}_2:\mathbb{H}^n\rightarrow\mathbb{H}^m$ be positive linear operators. If Assumption~\ref{assump:well-defined-N} is satisfied, then the function
    \begin{equation}\label{eqn:thm-main-barrier}
        (t, X) \mapsto -\log(S_g^{\mathcal{N}_1, \mathcal{N}_2}(X) - t) - \log\det(X),
    \end{equation}
    defined on $\mathbb{R}\times\mathbb{H}^n_{++}$ is an $(n+1)$-logarithmically homogeneous self-concordant barrier for the set
    \begin{equation}
        \closure\hypograph S_g^{\mathcal{N}_1, \mathcal{N}_2} = \closure \{ (t, X)\in\mathbb{R}\times\mathbb{H}^n_{++} : S_g^{\mathcal{N}_1, \mathcal{N}_2}(X) \geq t \},
    \end{equation}
    where $S_g^{\mathcal{N}_1, \mathcal{N}_2}:\mathbb{H}_{++}^n\rightarrow\mathbb{R}$ is defined in~\eqref{eqn:quasi-entropy-N}. Moreover, this barrier is optimal in the sense that any self-concordant barrier for $\closure\hypograph S_g^{\mathcal{N}_1, \mathcal{N}_2}$ has parameter at least $n+1$.
\end{theorem}
\begin{proof}
    See Section~\ref{subsec:main-proof}. 
\end{proof}
We note that when $\mathcal{N}_1$ and $\mathcal{N}_2$ both preserve positive definiteness, i.e., $\mathcal{N}_1(X)\succ0$ and $\mathcal{N}_2(X)\succ0$ for all $X\succ0$, then this result is a straightforward consequence of~\cite[Lemma 5.1.3(iii)]{nesterov1994interior}. However, it is nontrivial to extend this result to when $\mathcal{N}_1$ and $\mathcal{N}_2$ map to a face of the positive semidefinite cone, which will be the main effort in our proof for this theorem.

\diff{A straightforward consequence of Corollary~\ref{cor:main} is the following self-concordant barrier for the quantum conditional entropy function introduced in Example~\ref{exmp:structrure}.}
\begin{corollary}\label{cor:qce-cone}
    The function
    \begin{equation}\label{eqn:qce-barrier}
        (t,X) \mapsto -\log(t - S\divx{X}{\mathbb{I}\otimes\tr_1^{n,m}(X)}) - \log\det(X),
    \end{equation}
    defined on $\mathbb{R}\times\mathbb{H}^{nm}_{++}$ is an $(nm+1)$-logarithmically homogeneous self-concordant barrier of the quantum conditional entropy cone
    \begin{equation}
        \mathcal{K}_{\textnormal{qce}} = \closure \{ (t, X)\in\mathbb{R}\times\mathbb{H}^{nm}_{++} : t \geq S\divx{X}{\mathbb{I}\otimes\tr_1^{n,m}(X)} \}.
    \end{equation}
    Moreover, this barrier is optimal in the sense that any self-concordant barrier for $\mathcal{K}_{\textnormal{qce}}$ has parameter at least $nm+1$.
\end{corollary}
\diff{Later in Section~\ref{sec:exp}, we present additional extensions of Corollary~\ref{cor:main} which demonstrate how it can be applied to various applications arising in quantum information theory.}

\subsection{\diff{Compatibility} preliminaries} \label{subsec:concordant-prelim}

\diff{First}, we introduce a generalized notion of concavity with respect to the positive semidefinite cone, followed by the concept of compatibility of a function with respect to a domain.
\begin{definition}
    Let $\mathbb{V}$ be a finite-dimensional real vector space. A function $f:\domain f \subset \mathbb{V} \rightarrow \mathbb{H}^n$ with convex domain $\domain f \subset \mathbb{V}$ is $\mathbb{H}^n_+$-concave if for all $X,Y\in\domain f$ and $\lambda\in[0, 1]$, we have
    \begin{equation*}
        f(\lambda X + (1-\lambda) Y) \succeq \lambda f(X) + (1 - \lambda) f(Y).
    \end{equation*}
\end{definition}
\begin{definition}[{\cite[Definition 5.1.1]{nesterov1994interior}}]\label{defn:compatibility}
    Let $\mathbb{V}$ be a finite-dimensional real vector space, and let $f:\domain f \subset \mathbb{V} \rightarrow \mathbb{H}^n$ be a $C^3$, $\mathbb{H}^n_+$-concave function with open domain $\domain f \subset \mathbb{V}$. Then $f$ is $\beta$-compatible with the domain $\closure\domain f$ if there is a constant $\beta\geq0$ such that 
    \begin{equation*}
        \grad^3 f(x)[h, h, h] \preceq -3\beta\grad^2 f(x)[h, h],
    \end{equation*}
    for all $x\in\domain f$ and $h\in\mathbb{V}$ such that $x\pm h\in\closure\domain f$.
\end{definition}
\begin{remark}
    The proofs and results in~\cite{fawzi2023optimal} use an alternative definition of compatibility with respect to a barrier of a domain~\cite[Definition 5.1.2]{nesterov1994interior}. Although compatibility with respect to a domain is a stronger condition than compatibility with respect to a barrier of a domain~\cite[Remark 5.1.2]{nesterov1994interior}, all of the results from~\cite{fawzi2023optimal} that we use in this paper can be modified to use either definition for compatibility by using nearly identical arguments.
\end{remark}
We also introduce two important composition rules for compatibility with linear and affine maps from~\cite[Proposition 3.4]{fawzi2023optimal} and~\cite[Lemma 5.1.3(iii)]{nesterov1994interior}.
\begin{lemma}\label{lem:compatibility-composition}
    Let $\mathbb{V}$ and $\mathbb{V}'$ be finite-dimensional real vector spaces, let $f:\domain f \subset \mathbb{V} \rightarrow \mathbb{H}^n$ be a $C^3$, $\mathbb{H}^n_+$-concave function with open domain $\domain f \subset \mathbb{V}$ which is $\beta$-compatible with the domain $\closure\domain f$.
    \begin{enumerate}[label=(\roman*), ref=\ref{lem:compatibility-composition}(\roman*),leftmargin=25pt]
        \item \label{lem:compatibility-composition-i}Let $\mathcal{A}:\mathbb{H}^n\rightarrow\mathbb{H}^m$ be a positive linear map. Then $\mathcal{A}\circ f$ is $\beta$-compatible with the domain $\closure\domain f$. 
        \item \label{lem:compatibility-composition-ii}Let $\mathcal{B}:\mathbb{V}'\rightarrow\mathbb{V}$ be an affine map satisfying $\image\mathcal{B}\cap\domain f\neq\varnothing$. Then $f\circ\mathcal{B}$ is $\beta$-compatible with the domain
        \begin{equation*}
            \mathcal{B}^{-1}(\closure\domain f) \coloneqq \{ x\in\mathbb{V}' : \mathcal{B}(x) \in \closure\domain f \}.
        \end{equation*}
    \end{enumerate}
\end{lemma}
Compatibility is useful as it allows for a convenient method to construct self-concordant barriers of epigraphs of functions. This is summarized by~\cite[Theorem 3.3]{fawzi2023optimal}, which is a special case of~\cite[Theorem 5.4.4]{nesterov2018lectures}, and which we restate here for convenience.
\begin{lemma}\cite[Theorem 5.4.4]{nesterov2018lectures}\label{lem:compatibility-to-barrier}
    Let $\mathbb{V}$ be a finite-dimensional real vector space, and let $f:\domain f \subset \mathbb{V} \rightarrow \mathbb{H}^n$ be a $C^3$, $\mathbb{H}^n_+$-concave function with open domain $\domain f \subset \mathbb{V}$. If $F$ is a $\nu$-self-concordant barrier for $\closure\domain f$, and $f$ is $\beta$-compatible with $\closure\domain f$, then $(t, x)\mapsto -\log\det(f(x) - t) + \beta^3F(x)$ is an $(m+\beta^3\nu)$-self-concordant barrier for $\closure\hypograph f$.
\end{lemma}
Finally, we recall one of the results from~\cite{fawzi2023optimal} which establishes compatibility of the non-commutative perspective.
\begin{lemma}[{\cite[Proposition 3.5]{fawzi2023optimal}}]\label{lem:pg-compat}
    Let $g:(0, \infty)\rightarrow\mathbb{R}$ be an operator concave function, and let $P_g$ be its noncommutative perspective. Then $(X, Y)\in \mathbb{H}^{n_1}_{++}\times\mathbb{H}^{n_2}_{++}\mapsto P_g(X\otimes\mathbb{I}, \mathbb{I}\otimes\conj{Y})$ is $1$-compatible with the domain $\mathbb{H}^{n_1}_+\times\mathbb{H}^{n_2}_+$. 
\end{lemma}
Using this result together with appropriate applications of the composition rules from Lemma~\ref{lem:compatibility-composition} will form the backbone of our proof of Theorem~\ref{thm:main}.

\subsection{Proof of Theorem~\ref{thm:main}}\label{subsec:main-proof}

We begin with a simple lemma which helps us to characterize positive linear maps which map to singular matrices.
\begin{lemma}\label{lem:fr}
    Let $\mathcal{N}:\mathbb{H}^n\rightarrow\mathbb{H}^m$ be a positive linear map, and let $\mathcal{N}(\mathbb{I})$ have spectral decomposition
    \begin{equation*}
        \mathcal{N}(\mathbb{I}) = \begin{bmatrix} U & V\end{bmatrix} \begin{bmatrix} \Lambda & 0 \\ 0 & 0\end{bmatrix} \begin{bmatrix} U^\dag \\ V^\dag \end{bmatrix} = U\Lambda U^\dag,
    \end{equation*}
    where $U\in\mathbb{C}^{n\times r}$, $V\in\mathbb{C}^{n\times(n-r)}$ are isometry matrices, and $\Lambda\in\mathbb{R}^{r\times r}$ is a diagonal matrix with strictly positive diagonal entries. Then for any $X\in\mathbb{H}^n$, we have $\mathcal{N}(X) = U \hat{X} U^\dag$ for some $\hat{X}\in\mathbb{H}^r$. If, additionally, $X\succ0$, then $\hat{X}\succ0$.
\end{lemma}
\begin{proof}
    This follows from~\cite[Lemma 3.8]{hu2022robust} and~\cite[Theorem 6.6]{rockafellar1970convex}. 
\end{proof}
To prove Theorem~\ref{thm:main}, it would be convenient to use the composition rules outlined in Lemma~\ref{lem:compatibility-composition} together with the compatibility result from Lemma~\ref{lem:pg-compat}. However, when $\mathcal{N}_1$ or $\mathcal{N}_2$ map to singular matrices, we can no longer use a direct application of Lemma~\ref{lem:compatibility-composition-ii} as the intersection between a face of the positive semidefinite cone and the interior of the positive semidefinite cone is empty. Instead, we first show in the following theorem that we can always rewrite $S_g^{\mathcal{N}_1, \mathcal{N}_2}$ as an appropriate composition between the noncommutative perspective and positive linear maps which map to the interior of the positive semidefinite cone, after which the desired composition rules become applicable.
\begin{theorem}
    For a function $g:(0, \infty)\rightarrow\mathbb{R}$, let $P_g$ be its noncommutative perspective. Let $\mathcal{N}_k:\mathbb{H}^n\rightarrow\mathbb{H}^m$ for $k=1,2$ be positive linear operators such that
    \begin{equation*}
        \mathcal{N}_k(\mathbb{I}) = \begin{bmatrix} U_k & V_k \end{bmatrix} \begin{bmatrix} \Lambda_k & 0 \\ 0 & 0\end{bmatrix} \begin{bmatrix} U_k^\dag \\ V_k^\dag \end{bmatrix} = U_k\Lambda_k U_k^\dag,
    \end{equation*}
    where $U_k\in\mathbb{C}^{n\times r_k}$, $V_k\in\mathbb{C}^{n\times(n-r_k)}$ are isometry matrices, and $\Lambda_k\in\mathbb{R}^{r_k\times r_k}_{++}$ is a diagonal matrix with strictly positive diagonal entries. Let $\hat{\Psi}:\mathbb{H}^{n^2}\rightarrow\mathbb{R}$ be the unique linear operator satisfying $\hat{\Psi}[X\otimes \conj{Y}] = \tr[ U_1XU_1^\dag U_2YU_2^\dag ]$. Consider the function $h:\mathbb{H}^{n}_{++}\rightarrow\mathbb{R}$ defined by
    \begin{align*}
        h(X) &= \hat{\Psi} \bigl[P_g( U_1^\dag \mathcal{N}_1(X) U_1 \otimes\mathbb{I}, \mathbb{I}\otimes \conj{U_2^\dag\mathcal{N}_2(X)U_2}) \bigr].
    \end{align*}
    If Assumption~\ref{assump:well-defined-N} is satisfied, then for all $X\in\mathbb{H}^{n}_{++}$, we have
    \begin{equation*}
        S_g^{\mathcal{N}_1, \mathcal{N}_2}(X) = h(X) + \tr[X C],
    \end{equation*}
    where $C\in\mathbb{H}^n$ is some constant matrix dependent on $\mathcal{N}_1$, $\mathcal{N}_2$, and $g$.
\end{theorem}
\begin{proof}
    Consider any two positive definite matrices $A\in\mathbb{H}^{n_a}_{++}$ and $B\in\mathbb{H}^{n_b}_{++}$ with spectral decompositions $A=\sum_{i=1}^{n_a} a_i R_i$ and $B=\sum_{i=1}^{n_b} b_i S_i$. Since $A\otimes\mathbb{I}$ and $\mathbb{I}\otimes B$ commute, we have 
    \begin{align}\label{eqn:expanded-noncommutative-perspective}
        P_g(A\otimes\mathbb{I}, \mathbb{I}\otimes B) &= (A\otimes\mathbb{I})\,g(A^{-1} \otimes B) = \sum_{i=1}^{n_a}\sum_{j=1}^{n_b} a_ig\biggl(\frac{b_j}{a_i}\biggr) R_i \otimes S_j.
    \end{align}
    Now using Lemma~\ref{lem:fr}, for any $X\in\mathbb{H}^n_{++}$ we can always represent $\mathcal{N}_1(X)$ and $\mathcal{N}_2(X)$ using the spectral decompositions $\mathcal{N}_1(X) = \sum_{i=1}^{r_1} \lambda_i P_i$ and $\mathcal{N}_2(X) = \sum_{i=1}^{r_2} \mu_i Q_i$, where $\image P_i \subseteq \image U_1$ and $\lambda_i>0$ for all $i=1,\ldots,r_1$, and $\image Q_i \subseteq \image U_2$ and $\mu_i>0$ for all $i=1,\ldots,r_2$. Therefore, noting that $U_1^\dag \mathcal{N}_1(X) U_1$ and $U_2^\dag \mathcal{N}_2(X)U_2$ are positive definite and using~\eqref{eqn:expanded-noncommutative-perspective}, we obtain
    \begin{align*}
        h(X) &= \hat{\Psi} \biggl[ \sum_{i=1}^{r_1}\sum_{j=1}^{r_2} \lambda_i g \biggl( \frac{\mu_j}{\lambda_i}  \biggr) U_1^\dag P_i U_1\otimes \conj{U_2^\dag Q_jU_2} \biggr] = \sum_{i=1}^{r_1}\sum_{j=1}^{r_2} \lambda_i g \biggl( \frac{\mu_j}{\lambda_i}  \biggr) \tr[P_i Q_j].
    \end{align*}
    Now let $P^0=V_1V_1^\dag$ and $Q^0=V_2V_2^\dag$ be the projectors onto $\ker\mathcal{N}_1(X)$ and $\ker\mathcal{N}_2(X)$, respectively. Note that these projectors are independent of $X$, by Lemma~\ref{lem:fr}. Then using~\eqref{eqn:quasi-entropy-limit} and~\eqref{eqn:expanded-noncommutative-perspective} we can similarly show that
    \begin{align*}
        S_g^{\mathcal{N}_1, \mathcal{N}_2}(X) &= \lim_{\varepsilon\downarrow0} \Psi \bigl[P_g((\mathcal{N}_1(X) + \varepsilon\mathbb{I})\otimes\mathbb{I}, \mathbb{I}\otimes(\conj{\mathcal{N}_2(X)}+\varepsilon\mathbb{I})) \bigr]\\
        &= \lim_{\varepsilon\downarrow0} \biggl\{ \sum_{i=1}^{r_1}\sum_{j=1}^{r_2} \biggl[ (\lambda_i + \varepsilon) g\biggl(\frac{\mu_j + \varepsilon}{\lambda_i + \varepsilon} \biggr) \tr[P_i Q_j] \biggr] + \sum_{i=1}^{r_1} \biggl[ (\lambda_i + \varepsilon) g\biggl(\frac{\varepsilon}{\lambda_i + \varepsilon} \biggr) \tr[P_i Q^0] \biggr]  \\
        & \hphantom{ = \lim_{\varepsilon\downarrow0}} + \sum_{j=1}^{r_2} \biggl[(\mu_j + \varepsilon) \hat{g}\biggl(\frac{\varepsilon}{\mu_j + \varepsilon} \biggr) \tr[P^0 Q_j]\biggr] + \varepsilon g(1) \tr[P^0 Q^0] \biggr\} \\
        &= h(X) + \lim_{\varepsilon\downarrow0} \biggl\{ \sum_{i=1}^{r_1} \biggl[ (\lambda_i + \varepsilon) g\biggl(\frac{\varepsilon}{\lambda_i + \varepsilon} \biggr) \tr[P_i Q^0] \biggr]  + \sum_{j=1}^{r_2} \biggl[(\mu_j + \varepsilon) \hat{g}\biggl(\frac{\varepsilon}{\mu_j + \varepsilon} \biggr) \tr[P^0 Q_j] \biggr] \biggr\},
    \end{align*}
    where we have used the identity $xg(y/x)=y\hat{g}(x/y)$ in the second equality. For each scenario (i)--(iv) in Assumption~\ref{assump:well-defined-N}, the remaining limit is well-defined and reduces to a linear expression in $X$. We will only illustrate this for scenario (i), as similar arguments can be used for the other cases. For this scenario, we have that the columns of $V_2$ are all orthogonal to the columns of $U_1$, and therefore $\tr[P_i Q^0]=0$ for all $i=1,\ldots,r_1$. As the limit $\hat{g}(0^+)$ is also assumed to be finite, then
    \begin{align*}
        S_g^{\mathcal{N}_1, \mathcal{N}_2}(X) &= h(X) + \sum_{j=1}^{r_2} \mu_j \hat{g}(0^+) \tr[P^0 Q_j]= h(X) + \hat{g}(0^+)\tr[P^0 \mathcal{N}_2(X)],
    \end{align*}
    i.e., we have the desired form where $C=\hat{g}(0^+) \mathcal{N}_2^\dag(P^0)$, which is independent of $X$ since $P^0$ is independent of $X$. 
\end{proof}

Using this, we can prove the following compatibility result.
\begin{corollary}\label{cor:main-compatibility}
    Let $g:(0, \infty)\rightarrow\mathbb{R}$ be an operator concave function, and let $P_g$ be its noncommutative perspective. Let $\mathcal{N}_1:\mathbb{H}^n\rightarrow\mathbb{H}^m$ and $\mathcal{N}_2:\mathbb{H}^n\rightarrow\mathbb{H}^m$ be positive linear operators. If Assumption~\ref{assump:well-defined-N} is satisfied, then the function $S_g^{\mathcal{N}_1, \mathcal{N}_2}:\mathbb{H}_{++}^n\rightarrow\mathbb{R}$, as defined in~\eqref{eqn:quasi-entropy-N}, is $1$-compatible with the domain $\mathbb{H}^n_+$.
\end{corollary}
\begin{proof}
    First, using Lemmas~\ref{lem:compatibility-composition-i} and~\ref{lem:pg-compat} and recognizing that $\hat{\Psi}$ is a positive linear map, we know that 
    \begin{equation}\label{eqn:main-cor-proof-a}
        (X, Y)\mapsto \hat{\Psi}[P_g(X\otimes\mathbb{I}, \mathbb{I}\otimes\conj{Y})],
    \end{equation}
    defined on $\mathbb{H}^{r_1}_{++}\times\mathbb{H}^{r_2}_{++}$ is $1$-compatible with the domain $\mathbb{H}^{r_1}_+\times\mathbb{H}^{r_2}_+$. Next, we recognize from Lemma~\ref{lem:fr} that 
    \begin{equation*}
        \{(U_1^\dag \mathcal{N}_1(X) U_1, U_2^\dag \mathcal{N}_2(X) U_2): X\in\mathbb{H}^n\} \cap (\mathbb{H}^{r_1}_{++}\times\mathbb{H}^{r_2}_{++}) \neq \varnothing.
    \end{equation*}
    Therefore, we can apply Lemma~\ref{lem:compatibility-composition-ii} where $\mathcal{B}:\mathbb{H}^n\rightarrow\mathbb{H}^{r_1}\times\mathbb{H}^{r_2}, \mathcal{B}(X)=(U_1^\dag \mathcal{N}_1(X) U_1, U_2^\dag \mathcal{N}_2(X) U_2)$ and $f$ is~\eqref{eqn:main-cor-proof-a} to show that $h$ is $1$-compatible with the domain
    \begin{equation*}
        \mathcal{C} \coloneqq \{ X\in\mathbb{H}^n : U_k^\dag \mathcal{N}_k(X) U_k \in \mathbb{H}^{r_k}_+, \  \forall k=1,2 \}.
    \end{equation*}
    As $X\mapsto U_k^\dag \mathcal{N}_k(X) U_k$ are positive linear maps for $k=1,2$, we know that $\mathbb{H}^n_+\subseteq\mathcal{C}$. Therefore, it follows from the definition of compatibility that $h$ must also be $1$-compatible with the domain $\mathbb{H}^n_+$. By definition, compatibility only depends on the second and third derivatives of a function. Therefore it also follows that $S_g^{\mathcal{N}_1, \mathcal{N}_2}$, which only has an additional linear term, must also be $1$-compatible with the domain $\mathbb{H}^n_+$, as desired. 
\end{proof}
The proof for the main theorem is now a straightforward consequence of this compatibility result.
\begin{proof}[Proof of Theorem~\ref{thm:main}]
    The fact that~\eqref{eqn:thm-main-barrier} is a self-concordant barrier for $\closure \hypograph S_g^{\mathcal{N}_1, \mathcal{N}_2}$ follows from Corollary~\ref{cor:main-compatibility} and Lemma~\ref{lem:compatibility-to-barrier}, where we use the fact that $X\in\mathbb{H}^n_{++}\mapsto-\log\det(X)$ is an $n$-self-concordant barrier for the domain $\mathbb{H}^n_{+}$. Logarithmic homoegeneity with parameter $n+1$ can also be easily be directly confirmed from~\eqref{eqn:logarithmic-homogeneity}. Optimality of the barrier parameter follows from~\cite[Corollary 3.13]{fawzi2023optimal}. 
\end{proof}

\section{Exploiting structure}\label{sec:structure}

\begin{table*}
\centering
\caption{Summary of the approximate number of floating point operations required to solve linear systems with the Hessian of~\eqref{eqn:barrier-hat}, for various structures in the linear maps $\mathcal{G}:\mathbb{H}^n\rightarrow\mathbb{H}^{mp}$ and $\mathcal{H}:\mathbb{H}^n\rightarrow\mathbb{H}^{mp}$ that we can take advantage of. We compare the cost of a na\"ive matrix approach which does not use any structure in $\mathcal{G}$ and $\mathcal{H}$ to (i) if $\mathcal{G}$ and $\mathcal{H}$ both map to block diagonal matrices with $p$ blocks of size $m\times m$, (ii) if $\mathcal{G}$ and $\mathcal{H}$ are both rank $r^2$ linear maps, and (iii) if $\mathcal{G}(X)=X$ and $\mathcal{H}=\mathbb{I}\otimes\tr_1^{p,m}(X)$ model the quantum conditional entropy where $n=mp$. We also compare against the cost of solving a linear system with the Hessian of the barrier of the quantum relative entropy cone where matrices have dimension $mp$, i.e., the cone we would use if lifted the problem into a quantum relative entropy \diff{optimization problem}.} \label{table:flop-count}
\small
\begin{tabular*}{1\textwidth}{@{\extracolsep{\fill}}p{0.25cm}lccc@{\extracolsep{\fill}}}
\toprule
 && \textbf{Matrix construction} & \textbf{Matrix factorization} & \textbf{Matrix solve} \\ \midrule
&Na\"ive & $O(m^3p^3n^2 + m^2p^2n^4)$ & $O(n^6)$ & $O(n^4)$ \\
(i) &Block diag. & $O(m^3pn^2 + m^2pn^4)$ & $O(n^6)$ & $O(n^4)$ \\
(ii) &Low rank & $O(m^3p^3r^2 + m^2p^2r^4 + n^2r^4 + n^3r^2)$ & $O(r^6)$ & $O(n^3+n^2r^2+r^4)$ \\ 
(iii) &Quant. cond. entr. & $O(m^5p^3)$ & $O(m^6)$ & $O(m^3p^3 + m^4)$ \\
 &Quant. rel. entr. & $O(m^5p^5)$ & $O(m^6p^6)$ & $O(m^4p^4)$ \\
\bottomrule
\end{tabular*}
\end{table*}

As discussed in Section~\ref{sec:intro}, the main bottleneck when solving quantum relative entropy \diff{optimization problems} by using interior-point methods is in solving linear systems with the Hessian of the barrier function~\eqref{eqn:qre-barrier}. A na\"ive implementation of this step involves constructing and factoring the full Hessian matrix. In this section, we focus on how we can efficiently compute solve linear systems with the Hessian of the barrier function~\eqref{eqn:main-barrier} by taking advantage of structure in the linear maps $\mathcal{G}$ and $\mathcal{H}$. 
By performing a suitable block-elimination on the Hessian matrix of~\eqref{eqn:main-barrier}, we can show that the main computational effort is in solving linear systems with the Hessian of
\begin{equation}\label{eqn:barrier-hat}
    \hat{F}(X) = \zeta^{-1}S\divx{\mathcal{G}(X)}{\mathcal{H}(X)} - \log\det(X),
\end{equation}
where $\zeta=t-S\divx{\mathcal{G}(X)}{\mathcal{H}(X)}$ is treated as a constant. \diff{Details of this derivation can be found in Appendix~\ref{sec:barrier-hessian}}.
It is fairly straightforward to show that the Hessian matrix of $\hat{F}$ can be represented as
\begin{align}\label{eqn:hessian-naive}
    \grad^2\hat{F}(X) &= \zeta^{-1} \begin{bmatrix} \bm{\mathcal{G}}^\dag & \bm{\mathcal{H}}^\dag \end{bmatrix} \begin{bmatrix}
        -\grad^2 S(\mathcal{G}(X)) & \grad^2 S (\mathcal{H}(X)) \\
       \grad^2 S (\mathcal{H}(X)) & \grad^2 S_{\mathcal{G}(X)} (\mathcal{H}(X))
    \end{bmatrix} \begin{bmatrix} \bm{\mathcal{G}} \\ \bm{\mathcal{H}} \end{bmatrix}  - \grad^2\log\det(X),
\end{align}
where we use the bold letters $\bm{\mathcal{G}}$ and $\bm{\mathcal{H}}$ to denote the matrix representations of the linear maps $\mathcal{G}$ and $\mathcal{H}$, and we recall the definition of $S_{\mathcal{G}(X)}$ from~\eqref{eqn:funny-entropy}, where $\mathcal{G}(X)$ is treated as a constant. We refer to Appendix~\ref{sec:derivatives} for concrete descriptions of the Hessians that appear in~\eqref{eqn:hessian-naive}.

If we have no a priori information about $\mathcal{G}$ and $\mathcal{H}$, then the most straightforward method to solve linear systems with the Hessian matrix~\eqref{eqn:hessian-naive} is to construct the full matrix, then perform a Cholesky factor-solve procedure. Overall, the cost of forming the Hessian is $O(m^3n^2 + m^2n^4)$ flops, performing a Cholesky factorization costs $O(n^6)$ flops, and performing a single back- and forward-substitution using the Cholesky factorization costs $O(n^4)$. We compare this against solving linear systems with the Hessian of the log determinant function, as required in semidefinite programming. In this case, we do not need to construct nor Cholesky factor a Hessian matrix, and we apply the inverse Hessian map by applying a congruence transformation at a cost of $O(n^3)$ flops. 

However, if we have some information about $\mathcal{G}$ and $\mathcal{H}$ a priori, then there are certain structures we can exploit to either make constructing the Hessian cheaper, avoid having to build and store the whole Hessian in memory, or to improve the numerical stability of solving linear systems with the Hessian. We will explore three main categories of strategies in the remainder of this section, which we summarize the flop count complexities for in Table~\ref{table:flop-count}. We will later show in Section~\ref{sec:exp} how these techniques can be specifically tailored to solve various problems arising from quantum information theory. 

\subsection{Block diagonal structure} \label{subsec:blk-diag}
One structure we can take advantage of to make constructing the Hessian of~\eqref{eqn:barrier-hat} easier is if $\mathcal{G}$ and/or $\mathcal{H}$ map to block diagonal subspaces, i.e.,
\begin{subequations}\label{eqn:blk-diag-linear}
    \begin{align}
        \mathcal{G}(X) &= V \biggl( \bigoplus_{i=1}^{p_g} \mathcal{G}_i(X) \biggr) V^\dag = \begin{bmatrix} \vertbar && \vertbar \\ V_1 & \ldots & V_{p_g} \\ \vertbar && \vertbar \end{bmatrix} \begin{bmatrix} \mathcal{G}_1(X) && \\ & \ddots & \\ && \mathcal{G}_{p_g}(X) \end{bmatrix} \begin{bmatrix} \horzbar & V_1^\dag & \horzbar \\  & \vdots &  \\  \horzbar & V_{p_g}^\dag & \horzbar \end{bmatrix} \\
        \mathcal{H}(X) &= W \biggl( \bigoplus_{i=1}^{p_h} \mathcal{H}_i(X) \biggr) W^\dag = \begin{bmatrix} \vertbar && \vertbar \\ W_1 & \ldots & W_{p_h} \\ \vertbar && \vertbar \end{bmatrix} \begin{bmatrix} \mathcal{H}_1(X) && \\ & \ddots & \\ && \mathcal{H}_{p_h}(X) \end{bmatrix} \begin{bmatrix} \horzbar & W_1^\dag & \horzbar \\  & \vdots &  \\  \horzbar & W_{p_h}^\dag & \horzbar \end{bmatrix},
    \end{align}
\end{subequations}
where $\mathcal{G}_i:\mathbb{H}^n\rightarrow\mathbb{H}^{\hat{m}_{g_i}}$ and $V_i\in\mathbb{C}^{n\times \hat{m}_{g_i}}$ for $i=1,\ldots.p_g$, $\mathcal{H}_i:\mathbb{H}^n\rightarrow\mathbb{H}^{\hat{m}_{h_i}}$ and $W_i\in\mathbb{C}^{n\times \hat{m}_{h_i}}$ for $i=1,\ldots.p_h$, and $V$ and $W$ are unitary matrices. Then by using the property that for the trace of any spectral function $f$, we have
\begin{equation*}
    \tr\biggl[ f\biggl(U \biggl[\bigoplus_{i=1}^p X_i\biggr] U^\dag \biggr) \biggl] = \tr\biggl[\bigoplus_{i=1}^p f(X_i)\biggr] = \sum_{i=1}^p \tr[f(X_i)],
\end{equation*}
we can decompose quantum relative entropy into a sum of functions acting on smaller matrices, i.e.,
\begin{align}
    S\divx{\mathcal{G}(X)}{\mathcal{H}(X)} &= -S(\mathcal{G}(X)) + S_{\mathcal{G}(X)}(\mathcal{H}(X)) \nonumber \\
    &= -\sum_{i=1}^{p_g} S(\mathcal{G}_i(X)) + \sum_{i=1}^{p_h} S_{\hat{\mathcal{G}}_i(X)}(\mathcal{H}_i(X)), \label{eqn:block-diag-qre}
\end{align}
where $\hat{\mathcal{G}}_i(X) \coloneqq W_i\mathcal{G}(X) W_i^\dag$. Given this decomposition, building the Hessian just involves building the Hessian of each of these decomposed terms, which can be represented as
\begin{align}\label{eqn:hessian-block}
    \grad^2\hat{F}(X) &= -\zeta^{-1}\sum_{i=1}^{p_g} \bm{\mathcal{G}}^\dag_i \grad^2 S (\mathcal{G}_i(X)) \bm{\mathcal{G}}_i \nonumber \\
    &\hphantom{=} + \zeta^{-1}\sum_{i=1}^{p_h} \begin{bmatrix} \bm{\hat{\mathcal{G}}}_i^\dag & \bm{\mathcal{H}}_i^\dag \end{bmatrix} \begin{bmatrix}
        0 & \grad^2 S (\mathcal{H}_i(X)) \\
       \grad^2 S (\mathcal{H}_i(X)) & \grad^2 S_{\hat{\mathcal{G}}_i(X)} (\mathcal{H}_i(X))
    \end{bmatrix} \begin{bmatrix}  \bm{\hat{\mathcal{G}}}_i \\ \bm{\mathcal{H}}_i \end{bmatrix}  \- \grad^2\log\det(X).
\end{align}
By constructing the Hessian in this manner, we only incur a cost of $O(\hat{m}^3n^2p + \hat{m}^2n^4p)$ flops, where for simplicity, we have assumed $p_g=p_h=p$, and $\hat{m}_{g_i}=\hat{m}$ and $\hat{m}_{h_i}=\hat{m}$ for all $i=1,\ldots,p$. Noting that $m=\hat{m}p$, we see that we have saved a factor of approximately $p$ to $p^2$ flops in constructing the Hessian as compared to using the na\"ive expression~\eqref{eqn:hessian-naive}, depending on the relative sizes of $n$ and $m$.

\begin{remark}\label{rem:same-block}
     In the case that $\mathcal{G}$ and $\mathcal{H}$ share the same block diagonal structure, i.e., $p_g=p_h=p$, and $V_i=W_i$ for all $i=1,\ldots,p$ then~\eqref{eqn:block-diag-qre} can be expressed as
    \begin{align*}
        S\divx{\mathcal{G}(X)}{\mathcal{H}(X)} &= \sum_{i=1}^{p} S\divx{\mathcal{G}_i(X)}{\mathcal{H}_i(X)}.
    \end{align*}
\end{remark}

\begin{remark}\label{rem:classical}
    In the extreme case when $\mathcal{G}$ and $\mathcal{H}$ map to diagonal matrices, or, more generally, when $p_g=p_h=m$, and $\hat{m}_{g_i}=\hat{m}_{h_i}=1$ and $V_i=W_i$ for all $i=1,\ldots,m$, then we reduce to a classical relative entropy expression~\eqref{eqn:cre}.
\end{remark}

\subsubsection{Facial reduction}\label{subsubsec:facial-reduction}
When solving quantum relative entropy \diff{optimization problems} of the form~\eqref{eqn:qrep-base} by lifting to the full quantum relative entropy cone, an issue that arises is when $\mathcal{G}$ and/or $\mathcal{H}$ map to singular matrices. In this case, strict feasibility and Slater's condition no longer hold, which can lead to numerical issues with primal-dual interior-point algorithms~\cite{drusvyatskiy2017many}.  Many works such as~\cite{hu2022robust} use facial reduction to recover strict feasibility. 

Alternatively, we can recover strict feasibility by modelling the problem directly using the cone~\eqref{eqn:main-barrier}. We can also take advantage of knowledge of the fact that $\mathcal{G}$ and/or $\mathcal{H}$ map to singular matrices to compute the Hessian matrix of the barrier function more easily, which we demonstrate as follows. In the notation of~\eqref{eqn:blk-diag-linear} and without loss of generality, let us assume that $p_g=p_h=2$, and $\mathcal{G}_{2}(X)=0$ and $\mathcal{H}_{2}(X)=0$ for all $X\in\mathbb{H}^n$. Recalling that we assumed $\mathcal{G}(X) \ll \mathcal{H}(X)$ for all $X\in\mathbb{H}^n_{++}$, Lemma~\ref{lem:fr} implies that $\hat{\mathcal{G}}_{2}(X)=0$ for all $X\in\mathbb{H}^n$. Using $0\log(0)=0$, we can simplify~\eqref{eqn:block-diag-qre} to
\begin{align}
    S\divx{\mathcal{G}(X)}{\mathcal{H}(X)} &= -S(V_1^\dag\mathcal{G}(X)V_1) + S_{W_1^\dag\hat{\mathcal{G}}(X)W_1}(W_1^\dag\mathcal{H}(X)W_1) \nonumber \\
    &= -S(\mathcal{G}_1(X)) + S_{\hat{\mathcal{G}}_1(X)}(\mathcal{H}_1(X)),
\end{align}
i.e., we drop the terms corresponding to the kernels of $\mathcal{G}(\mathbb{I})$ and $\mathcal{H}(\mathbb{I})$. The corresponding Hessian matrix $\grad^2\hat{F}(X)$ is a straightforward modification of~\eqref{eqn:hessian-block}. Overall, we have reduced the size of the matrices we need to consider, which can significantly simplify computations if the dimensions of the kernels of $\mathcal{G}(\mathbb{I})$ and $\mathcal{H}(\mathbb{I})$ are large. Note that this does not preclude the possibility that $\mathcal{G}_1$ and $\mathcal{H}_1$ have further block diagonal structure that we can exploit in a similar way as~\eqref{eqn:block-diag-qre}. 


\subsection{Low-rank structure} \label{subsec:low-rank}
Although the block diagonal structure allows us to construct the Hessian more efficiently, it is often more desirable if we can solve linear systems with the Hessian without having to construct and factor the entire Hessian matrix. We can do this if $\mathcal{G}$ and/or $\mathcal{H}$ are low rank linear maps, in which case we can apply variants of the matrix inversion lemma discussed in Appendix~\ref{appdx:matrix-inversion-lemma} to solve linear systems with the Hessian more efficiently. More concretely, let us assume that we can decompose $\mathcal{G}$ and $\mathcal{H}$ as
\begin{subequations}
    \begin{align}
        \mathcal{G} &= \mathcal{G}_2 \circ \mathcal{G}_1 \\
        \mathcal{H} &= \mathcal{H}_2 \circ \mathcal{H}_1,
    \end{align}
\end{subequations}
where $\mathcal{G}_1:\mathbb{H}^n\rightarrow\mathbb{H}^{r_g}$, $\mathcal{G}_2:\mathbb{H}^{r_g}\rightarrow\mathbb{H}^m$, $\mathcal{H}_1:\mathbb{H}^n\rightarrow\mathbb{H}^{r_h}$, and $\mathcal{H}_2:\mathbb{H}^{r_h}\rightarrow\mathbb{H}^m$, and $r_g \ll n$ and $r_h \ll n$, respectively. In this case, we can write the Hessian of $\hat{F}$ as
\begin{align}\label{eqn:hessian-low-rank}
    \grad^2 \hat{F}(X) &= \zeta^{-1}\begin{bmatrix} \bm{\mathcal{G}}_1^\dag & \bm{\mathcal{H}}_1^\dag \end{bmatrix} \begin{bmatrix}
        -\bm{\mathcal{G}}_2^\dag\grad^2 S (\mathcal{G}(X)) \bm{\mathcal{G}}_2 & \bm{\mathcal{G}}_2^\dag\grad^2 S (\mathcal{H}(X)) \bm{\mathcal{H}}_2\\
       \bm{\mathcal{H}}_2^\dag\grad^2 S (\mathcal{H}(X)) \bm{\mathcal{G}}_2 & \bm{\mathcal{H}}_2^\dag\grad^2 S_{\mathcal{G}(X)} (\mathcal{H}(X)) \bm{\mathcal{H}}_2
    \end{bmatrix} \begin{bmatrix} \bm{\mathcal{G}}_1 \\ \bm{\mathcal{H}}_1 \end{bmatrix} - \grad^2\log\det(X).
\end{align}
Constructing the full Hessian in this form costs $O(m^3r^2 + m^2r^4 + n^2r^4 + n^4r^2)$ flops, where for simplicity we assumed $r_g=r_h=r$. When $r\ll n$, this is already cheaper than constructing the full Hessian using the na\"ive approach~\eqref{eqn:hessian-naive}. However, we can gain further computational advantages by treating the Hessian as a rank-$2r^2$ perturbation of the matrix $\grad^2\log\det(X)$, which we can efficiently compute inverse Hessian products with using~\eqref{eqn:logdet-hess-inv}. Therefore, we can solve linear systems with $\grad^2 \hat{F}(X)$ by using the matrix inversion lemma~\eqref{eqn:matrix-inversion-lemma-a}. In this case, we only need to construct a smaller Schur complement matrix which takes $O(m^3r^2 + m^2r^4 + n^2r^4 + n^3r^2)$ flops, and factor this matrix which now only costs $O(r^6)$ flops. 

\begin{example}\label{exmp:low-rank-small}
    A common example of a low-rank linear map is when $\mathcal{G}$ and $\mathcal{H}$ map to small matrices, i.e., $m \ll n$. In this case, $r_h=r_g=m$ and $\mathcal{G}_2=\mathcal{H}_2=\mathbb{I}$.
\end{example}

\begin{example}\label{exmp:low-rank-submatrix}
    Another example of low-rank linear maps is when $\mathcal{G}$ and $\mathcal{H}$ only act on a small principal submatrix of the input matrix. In this case, we can let $\mathcal{G}_1$ and $\mathcal{H}_1$ be the linear operators which extract the corresponding principal submatrices.
\end{example}

\begin{remark}\label{rem:efficient-schur}
    When constructing the Schur complement matrix for~\eqref{eqn:hessian-low-rank}, we need to construct a matrix of the form
    \begin{equation*}
        \begin{bmatrix} \bm{\mathcal{G}}_1 \\ \bm{\mathcal{H}}_1 \end{bmatrix} \grad^2\log\det(X) \begin{bmatrix} \bm{\mathcal{G}}_1^\dag & \bm{\mathcal{H}}_1^\dag \end{bmatrix},
    \end{equation*}
    where we recall that $\grad^2\log\det(X)$ is the Hessian corresponding to the log determinant, i.e., the barrier for the positive semidefinite component of the cone. While we discussed briefly in Section~\ref{sec:intro} how sparse and low-rank structure was difficult to exploit to perform Hessians products of general spectral functions more efficiently, there are well-known techniques to exploit these same structures for the Hessian products of the log determinant~\cite{fujisawa1997exploiting,benson2008dsdp5}. Therefore, if $\mathcal{G}_1$ or $\mathcal{H}_1$ have these sparse or low-rank structures (such as those described in Example~\ref{exmp:low-rank-submatrix}), then these same techniques can be used to construct this component of the Schur complement matrix more efficiently. 
\end{remark}

\subsubsection{Identity operator} \label{subsubsec:identity-low-rank}
A special instance of this low-rank structure occurs when $\mathcal{G}=\mathbb{I}$ (or is a congruence by any unitary matrix) and $\mathcal{H}=\mathcal{H}_2\circ\mathcal{H}_1$ is low rank, where $\mathcal{H}_1:\mathbb{H}^n\rightarrow\mathbb{H}^{r}$, and $\mathcal{H}_2:\mathbb{H}^{r}\rightarrow\mathbb{H}^m$, and $r \ll n$. First, let us represent the Hessian matrix as
\begin{equation}\label{eqn:id-low-rank}
    \grad^2 \hat{F}(X) = \zeta^{-1} \begin{bmatrix} \mathbb{I} & \bm{\mathcal{H}}_1^\dag \end{bmatrix} \begin{bmatrix}
        -\grad^2 (S + \,\zeta\log\det) (X) & \grad^2 S (\mathcal{H}(X))\bm{\mathcal{H}}_2 \\
       \bm{\mathcal{H}}_2^\dag\grad^2 S (\mathcal{H}(X)) & \bm{\mathcal{H}}_2^\dag\grad^2 S_{X} (\mathcal{H}(X)) \bm{\mathcal{H}}_2
    \end{bmatrix} \begin{bmatrix} \mathbb{I} \\ \bm{\mathcal{H}}_1 \end{bmatrix}. 
\end{equation}
This matrix can be interpreted as a rank-$2r^2$ perturbation of the matrix $\grad^2 (S + \,\zeta\log\det) (X)$. We can efficiently solve linear systems with by using Remark~\ref{rem:easy-inverse-hessian}, where $f(x)=-x\log(x)+\zeta/x$. Therefore, by using the variant of the matrix inversion lemma from Lemma~\ref{lem:matrix-inversion-lemma-c} on~\eqref{eqn:id-low-rank}, the cost of constructing the required Schur complement matrix is still $O(m^3r^2 + m^2r^4 + n^2r^4 + n^3r^2)$ flops, and the matrix factorization step now only costs $O(r^6)$ flops. That is, we can still apply a variant of the matrix inversion lemma to compute the inverse Hessian product efficiently, despite $\mathcal{G}$ being full rank, due to the sum of the matrices $\grad^2 S (X)$ and $\grad^2 \log\det (X)$ being easily invertible. 

\subsection{Difference of entropies}\label{subsec:diff-entropies}
In some instances, a certain block diagonal structure of the linear map $\mathcal{H}$ allows us to decompose quantum relative entropy into a difference of quantum entropies. 
\begin{lemma}\label{lem:partial-trace-entropy}
    For all $X\in\mathbb{H}^{nm}_+$, the identity $S_X(\mathbb{I}\otimes\tr_1^{n,m}[X]) = S(\tr_1^{n,m}(X))$ is satisfied.
\end{lemma}
\begin{proof}
    This follows from
    \begin{align*}
        S_X(\mathbb{I}\otimes\tr_1^{n,m}(X)) &= -\inp{X}{\log(\mathbb{I}\otimes\tr_1^{n,m}(X))} \\
        &= -\inp{X}{\mathbb{I}\otimes\log(\tr_1^{n,m}(X))} \\
        &= -\inp{\tr_1^{n,m}(X)}{\log(\tr_1^{n,m}(X))} = S(\tr_1^{n,m}(X)),
    \end{align*}
    where the second equality uses the fact that spectral functions act on each block of a block diagonal matrix independently, and the third equality uses the fact that the adjoint of the partial trace is the Kronecker product, i.e., $\inp{Y}{\tr_1^{n,m}(X)}=\inp{\mathbb{I}\otimes Y}{X}$ for all $X\in\mathbb{H}^{nm}$ and $Y\in\mathbb{H}^m$. 
\end{proof}
We now show how we can generalize this result to capture a larger class of linear operators which can be decomposed into a difference of entropies. Consider a linear operator $\mathcal{H}$ which consists of a direct sum of the partial trace expressions from Lemma~\ref{lem:partial-trace-entropy}, i.e.,
\begin{align}
    \mathcal{H}(X) &= V \bigoplus_{i=1}^{p} \biggl( (\mathcal{H}_i^\dag \circ \mathcal{H}_i \circ \mathcal{G}_i)(X) \biggl) V^\dag \nonumber \\
    &= 
    \begin{bmatrix} \vertbar && \vertbar \\ V_1 & \ldots & V_{p} \\ \vertbar && \vertbar \end{bmatrix} 
    \setlength{\arraycolsep}{0pt}
    \begin{bmatrix} (\mathcal{H}_1^\dag \circ \mathcal{H}_1 \circ \mathcal{G}_1)(X) && \\ & \ddots & \\ && (\mathcal{H}_p^\dag \circ \mathcal{H}_p \circ \mathcal{G}_p)(X) \end{bmatrix} 
    \begin{bmatrix} \horzbar & V_1^\dag & \horzbar \\  & \vdots &  \\  \horzbar & V_{p}^\dag & \horzbar \end{bmatrix},
\end{align}
where $\mathcal{H}_i(X)=\tr_1^{a_i,b_i}(X)$ are partial trace operators for all $i=1,\ldots,p$. Additionally, $V_i\in\mathbb{C}^{n\times a_ib_i}$ and
$\mathcal{G}_i(X)=V_i^\dag \mathcal{G}(X) V_i$ for $i=1,\ldots,p$, and $V$ is a unitary matrix. Then we can express the quantum relative entropy function as a sum of quantum entropy terms
\begin{align*}
    S\divx{\mathcal{G}(X)}{\mathcal{H}(X)} &= -S(\mathcal{G}(X)) + \sum_{i=1}^p S_{\mathcal{G}_i(X)}((\mathcal{H}_i^\dag \circ \mathcal{H}_i \circ \mathcal{G}_i)(X)) \\
    &= -S(\mathcal{G}(X)) + \sum_{i=1}^p S((\mathcal{H}_i \circ \mathcal{G}_i)(X)),
\end{align*}
where we use~\eqref{eqn:block-diag-qre} for the first equality, and Lemma~\ref{lem:partial-trace-entropy} for the last equality. The main advantage of this reformulation is that we can now express the Hessian as
\begin{align}\label{eqn:diff-entropy-hess}
    \grad^2 \hat{F}(X) &=  -\zeta^{-1}\bm{\mathcal{G}}^\dag \grad^2 S (\mathcal{G}(X)) \bm{\mathcal{G}} + \zeta^{-1}\sum_{i=1}^{p} \biggl(\bm{\mathcal{G}}_i^\dag \bm{\mathcal{H}}_i^\dag \grad^2 S ((\mathcal{H}_i \circ \mathcal{G}_i)(X)) \bm{\mathcal{H}}_i \bm{\mathcal{G}}_i \biggr) - \grad^2 \log\det (X),
\end{align}
which can be interpreted as a simplification of the Hessian previously found in~\eqref{eqn:hessian-block}. This allows us to avoid computing the second divided differences matrix associated with $\grad^2 S_C$ (see Lemma~\ref{eqn:lem-tr-c-f-derivative}), which should help improve the numerical stability of constructing the Hessian matrix. 

\begin{example}\label{exmp:conditional}
    An important example of a linear map satisfying this property is the quantum conditional entropy
    \begin{equation*}
        X \mapsto S\divx{X}{\mathbb{I} \otimes \tr_1^{n,m}(X)},
    \end{equation*}
    defined over $\mathbb{H}^{nm}_{+}$. This corresponds to the case when $p=1$, $\mathcal{G}(X)=X$, and $\mathcal{H}_1(X)=\tr_1^{n,m}(X)$. For the quantum conditional entropy, the Hessian~\eqref{eqn:diff-entropy-hess} can be simplified to 
    \begin{equation*}
        \grad^2 \hat{F}(X) = -\grad^2 (\zeta^{-1} S + \log\det) (X) + \bm{\mathcal{H}}^\dag_1 \grad^2 S (\mathcal{H}_1(X)) \bm{\mathcal{H}}_1.
    \end{equation*}
    This is a similar scenario to that discussed in Section~\ref{subsubsec:identity-low-rank}, and therefore we can use the matrix inversion lemma to efficiently compute inverse Hessian products for the barrier of the quantum conditional entropy cone introduced in Corollary~\ref{cor:qce-cone}. Note that the small matrix corresponding to the low rank perturbation matrix is positive definite and easily invertible, unlike in, e.g.,~\eqref{eqn:hessian-low-rank}. Therefore, we can use the symmetric variant of the matrix inversion lemma~\eqref{eqn:matrix-inversion-lemma-b}, and we can perform a Cholesky factorization of the Schur complement matrix rather than an LU factorization. Overall, it costs $O(n^3m^5)$ flops to construct the Schur complement, $O(m^6)$ flops to Cholesky factor the Schur complement, and $O(n^3m^3+m^4)$ flops to perform a linear solve using this Cholesky factorization.
\end{example}

\begin{example}\label{exmp:pinching}
    Another example is when $\mathcal{H}$ is the pinching map~\eqref{eqn:pinching-map}, which zeros out all off-diagonal blocks of a block matrix. This corresponds to the case when $a_1=\ldots=a_p=1$, i.e., the partial traces act like the identity. More concretely, we have
    \begin{align*}
        S\divx{\mathcal{G}(X)}{\mathcal{Z}(\mathcal{G}(X))} &= -S(\mathcal{G}(X)) + \sum_{i=1}^p S(\mathcal{G}_i(X)),
    \end{align*}
    where $\mathcal{G}_i(X)$ represents the $i$-th diagonal block of $\mathcal{G}(X)$. 
\end{example}

\section{Numerical experiments}\label{sec:exp}

In this section, we present numerical results for solving \diff{a variety} of quantum relative entropy \diff{optimization problems} \diff{arising in quantum information theory} to demonstrate how the techniques presented in this paper can be utilized. For each experiment, we show how 1) we can use tailored cones and barriers from Section~\ref{sec:opt-barrier} to model the quantum relative entropy \diff{optimization problem}, and 2) how to use techniques from Section~\ref{sec:structure} to efficiently solve linear systems with the Hessian matrix of the barrier.


All experiments are run in Julia using Hypatia~\cite{coey2023performance}, which already provides support for the quantum relative entropy cone, quantum entropy cone, and classical relative entropy cone. All experiments are run using Hypatia's default settings, and on an Intel i5-11700 CPU with 32GB of RAM. We report the total number of iterations and time taken to solve each problem to the desired tolerance (Hypatia uses a default relative gap tolerance and primal-dual feasibility tolerance of $1.5\times10^{-8}$), excluding time spent in preprocessing steps. We also report the average time taken to compute a single interior-point step. We use ``OOM'' to indicate instances where there is insufficient computational memory to solve a problem. Whenever problems are modelled using a constraint of the form $h-Gx\in\mathcal{K}$ where $G\neq-\mathbb{I}$, we use Hypatia's default system solver, which is based on the method from~\cite[Section 10.3]{vandenberghe2010cvxopt}. Otherwise, we use a custom-implemented Newton system solver based on the block elimination method discussed in Appendix~\ref{appdx:pdipm}. We present more details about how each experiment is modelled in the standard conic program form accepted by Hypatia in Appendix~\ref{appdx:model}.

\diff{For the interested reader, we also provide a brief comparison between our proposed method and prior techniques which have been used to solve the quantum relative entropy optimization problems in each application we explore.}

\subsection{Quantum key distribution}\label{subsec:qkd-experiments}

\diff{Quantum key distribution is a scheme for generating and sharing a private key, which allows two parties to communicate information securely. The security of quantum key distribution protocols depends on a quantity called the key rate, which was originally given by the Devetak-Winter formula~\cite{devetak2005distillation}. This was later simplified in~\cite{coles2016numerical,winick2018reliable} to the following convex optimization problem}
\begin{subequations}\label{eqn:qkd}
    \begin{align}
        \minimize_{X\in\mathbb{H}^n} \quad & S\divx{\mathcal{G}(X)}{\mathcal{Z}(\mathcal{G}(X))} + p_{\textnormal{pass}}\delta_{\textnormal{EC}} \\
        \subjto \quad & \mathcal{A}(X) = b \\
        & X \succeq 0,
    \end{align}
\end{subequations}
\diff{where $\mathcal{G}:\mathbb{H}^n\rightarrow\mathbb{H}^{mr}$ is a positive linear map related to a particular description of a quantum key distribution protocol. Usually, $\mathcal{G}$ is given in the form}
\begin{equation}
    \mathcal{G}(X) = \sum_{i=1}^l K_i X K_i^\dag,
\end{equation}
\diff{where $K_i:\mathbb{C}^n\rightarrow\mathbb{C}^{mr}$ for $i=1,\ldots,l$, and are known as the Kraus operators associated with $\mathcal{G}$. Additionally, $\mathcal{Z}:\mathbb{H}^{mr}\rightarrow\mathbb{H}^{mr}$ is another positive linear map known as the pinching map, which maps all off-diagonal blocks of a matrix to zero for a given block structure. Concretely, we can represent this as}
\begin{equation}\label{eqn:pinching-map}
    \mathcal{Z}(Y) = \sum_{i=1}^r (e_ie_i^\top \otimes \mathbb{I}_{m}) Y (e_ie_i^\top \otimes \mathbb{I}_{m}),
\end{equation}
\diff{where $e_i\in\mathbb{C}^r$ are the standard $r$-dimensional basis vectors for $i=1,\ldots,r$. Finally, $\mathcal{A}:\mathbb{H}^n\rightarrow\mathbb{R}^p$ and $b\in\mathbb{R}^p$ encode a set of experimental constraints $\mathcal{A}(X)=b$, and $p_{\textnormal{pass}}$ and $\delta_{\textnormal{EC}}$ are constants dependent on the quantum key distribution protocol.}


\diff{To compute the quantum key rate}, we can directly model~\eqref{eqn:qkd} using a suitable application of the cone and barrier presented in Corollary~\ref{cor:main} (see~\eqref{eqn:dprbb84-dpr} for a concrete representation). To construct the Hessian matrix of this barrier, we first take advantage of the pinching map $\mathcal{Z}$ to decompose the quantum relative entropy into a difference of quantum entropies, as seen in Example~\ref{exmp:pinching}. Additionally, the image of $\mathcal{G}$ consistes of (typically) large but singular matrices with a large common kernel. Therefore, we can use the technique discussed in Section~\ref{subsubsec:facial-reduction} to reduce the size of the matrices we deal with. By defining $V$ and $W_i$ as the isometry matrices with orthonormal columns living in the image of matrices $\mathcal{G}(\mathbb{I})$ and $\mathcal{G}_i(\mathbb{I})$ for $i=1,\ldots,p$, respectively, we obtain
\begin{align}\label{eqn:qkd-hessian}
    S\divx{\mathcal{G}(X)}{\mathcal{Z}(\mathcal{G}(X))} &= -S(V^\dag\mathcal{G}(X)V) + \sum_{i=1}^p S(W_i^\dag\mathcal{G}_i(X)W_i).
\end{align}
The Hessian matrix can then be computed using a straightforward modification of~\eqref{eqn:diff-entropy-hess}. Similar techniques were used in~\cite{hu2022robust} to compute quantum key rates.

We now focus on the two particular quantum key distribution protocols to see how knowledge about $\mathcal{G}$ can further simplify computation of the inverse Hessian product, and present numerical results for each of these problems. \diff{Note we choose to study these protocols as their problem dimensions are parameterized, and are therefore the most computationally interesting examples out of the protocols studied in~\cite{hu2022robust}.}

\subsubsection{Discrete-phase-randomized protocol}
\diff{\sloppy We first introduce the discrete-phase-randomized variant of the Bennett-Brassard 1984 (dprBB84) protocol~\cite{cao2015discrete}. This protocol is parameterized by a number of global phases $c$ used to generate the private key, and a probability $p$ that a particular measurement basis is chosen. Given these parameters, the Kraus operators of $\mathcal{G}:\mathbb{H}^{12c}\rightarrow\mathbb{H}^{48c}$ are given by}
\begin{subequations}\label{eqn:dprbb84-kraus}
    \begin{align}
        K_1 = \sqrt{p} &\mleft[ \begin{bmatrix} 1 \\ 0 \end{bmatrix} \otimes \biggl(\bigoplus_{i=1}^c \Pi_1 \biggr) + \begin{bmatrix} 0 \\ 1 \end{bmatrix} \otimes \biggl(\bigoplus_{i=1}^c \Pi_2 \biggr) \mright] \otimes \begin{bmatrix}
            1 & & \\ & 1 & \\ & & 0
        \end{bmatrix} \otimes \begin{bmatrix} 1 \\ 0 \end{bmatrix} \\ 
        K_2 = \sqrt{1 - p} &\mleft[ \begin{bmatrix} 1 \\ 0 \end{bmatrix} \otimes \biggl(\bigoplus_{i=1}^c \Pi_3 \biggr) + \begin{bmatrix} 0 \\ 1 \end{bmatrix} \otimes \biggl(\bigoplus_{i=1}^c \Pi_4 \biggr) \mright] \otimes \begin{bmatrix}
            1 & & \\ & 1 & \\ & & 0
        \end{bmatrix} \otimes \begin{bmatrix} 0 \\ 1 \end{bmatrix},
    \end{align}   
\end{subequations}
\diff{where $\Pi_i\in\mathbb{H}^4$ is defined as $\Pi_i=e_ie_i^\top\in\mathbb{H}^4$ for $i=1,2,3,4$. The pinching map $\mathcal{Z}:\mathbb{H}^{48c}\rightarrow\mathbb{H}^{48c}$ is defined to act on a $2\times2$ block matrix, i.e., $r=2$. See~\cite[Section IV.D]{george2021numerical} for additional details of this protocol.}

In Appendix~\ref{appdx:dprBB84} we show that both $\mathcal{G}$ and $\mathcal{Z}\circ\mathcal{G}$ have block diagonal structures~\eqref{eqn:dpr-g} and~\eqref{eqn:dpr-z}, respectively. Applying the techniques for block diagonal structures from Section~\ref{subsec:blk-diag} and the facial reduction techniques from Section~\ref{subsubsec:facial-reduction}, we therefore have
\begin{equation}\label{eqn:dpr-a}
    S\divx{\mathcal{G}(X)}{\mathcal{Z}(\mathcal{G}(X))} = \sum_{i=1}^2 \biggl( -S(\mathcal{G}_i(X)) + S(\mathcal{Z}_1(\mathcal{G}_i(X))) + S(\mathcal{Z}_2(\mathcal{G}_i(X))) \biggr),
\end{equation}
where $\mathcal{G}_1:\mathbb{H}^{12c}\rightarrow\mathbb{H}^{4c}$, $\mathcal{G}_2:\mathbb{H}^{12c}\rightarrow\mathbb{H}^{4c}$, $\mathcal{Z}_1:\mathbb{H}^{4c}\rightarrow\mathbb{H}^{2c}$ and $\mathcal{Z}_2:\mathbb{H}^{4c}\rightarrow\mathbb{H}^{2c}$ are defined in Appendix~\ref{appdx:dprBB84}. To summarize, we have facially reduced $\mathcal{G}(X)$ and $\mathcal{Z}(\mathcal{G}(X))$ from size $48c\times48c$ matrices to size $8c\times8c$ matrices. We then further decomposed $\mathcal{G}(X)$ into two blocks of size $4c\times4c$, and $\mathcal{Z}(\mathcal{G}(X))$ into four blocks of size $2c\times2c$. As these reduced matrices are relatively small compared to $X$, then, as discussed in Example~\ref{exmp:low-rank-small}, we can also exploit this low-rank structure to compute the inverse Hessian product more efficiently. We can see this by expressing the Hessian matrix as
\begin{equation}\label{eqn:dpr-hessian}
    \grad^2 \hat{F}(X) = \zeta^{-1} \begin{bmatrix}
        \bm{\mathcal{G}}_1^\dag & \bm{\mathcal{G}}_2^\dag
    \end{bmatrix} \begin{bmatrix}
        M_1 \\ & M_2
    \end{bmatrix}  \begin{bmatrix}
        \bm{\mathcal{G}}_1 \\ \bm{\mathcal{G}}_2
    \end{bmatrix} - \grad^2\log\det(X),
\end{equation}
where
\begin{equation*}
    M_i = -\grad^2 S(\mathcal{G}_i(X)) + \bm{\mathcal{Z}}_1^\dag \grad^2 S(\mathcal{Z}_1(\mathcal{G}_i(X))) \bm{\mathcal{Z}}_1 + \bm{\mathcal{Z}}_2^\dag \grad^2 S(\mathcal{Z}_2(\mathcal{G}_i(X))) \bm{\mathcal{Z}}_2,
\end{equation*}
for $i=1,2$. Here, we have a rank-$32c^2$ perturbation of an easily invertible $144c^2\times144c^2$-dimensional log determinant Hessian matrix, and therefore we can use the matrix inversion lemma~\eqref{eqn:matrix-inversion-lemma-a} to solve linear systems with this Hessian matrix. Additionally, as $\mathcal{G}_1$ and $\mathcal{G}_2$ are defined to extract a principal submatrix from the input matrix, they are therefore highly structured and the discussion in Remark~\ref{rem:efficient-schur} applies to allow us compute the Schur complement matrix more efficiently. We denote this proposed method where we model using a suitably parameterized cone from Corollary~\ref{cor:main}, as shown in~\eqref{eqn:dprbb84-dpr}, and use the aforementioned technique to solve linear systems with the Hessian of the barrier, as DPR. 

\paragraph{Experimental setup}
We run experiments for dprBB84 protocols for various numbers of global phases $c$ and a probability $p=0.5$ of using a particular measurement basis. We simulate the experimental constraint data $\mathcal{A}$ and $b$ using a source intensity of $\alpha=0.14$ and transmission channel length of $L=\SI{30}{\meter}$. All other experimental parameters are the same as those used in~\cite{hu2022robust}. We benchmark DPR against two other baseline methods. First, we again model the problem using a suitably parameterized cone from Corollary~\ref{cor:main}, as shown in~\eqref{eqn:dprbb84-dpr}, but instead solve linear systems with the Hessian of the barrier by constructing the Hessian matrix~\eqref{eqn:dpr-hessian} and performing a Cholesky factor-solve procedure (QKD). Second, we model the problem using the full quantum relative entropy cone (QRE), as shown in~\eqref{eqn:dprbb84-qre}, which we preprocess with the technique~\cite[Lemma 6.2]{karimi2023efficient} to reduce the size of the matrices we need to work with. Our experimental results are summarized in Table~\ref{tab:qkd-dpr}. 

\paragraph{Discussion}
Overall, we observe that using a cone and barrier specialized for the quantum key distribution problem can yield noticeable improvements to solve times. 
We see that the QKD method improves solve times by approximately a factor of $2$ over the QRE method. By further taking advantage of our knowledge that the Hessian can be interpreted as a low rank perturbation of the log determinant Hessian (i.e., DPR), we observe a further $2$ to $5$ times speedup in solve time. We note that the DPR technique seemed to be less numerically stable as the magnitudes of the residuals from solving the Newton equations solved to were observed to be larger, however the method was still able to converge to the desired tolerance. 
We compare our results to those presented in~\cite{hu2022robust}, where the authors present results for the dprBB84 protocol only for problems up to $c=5$, whereas we present experiments for problems up to $c=10$. For problems of dimension $c=4$, the reported computation times in~\cite{hu2022robust} are all at least $\SI{300}{\second}$, which is slower than all methods that we present in Table~\ref{tab:qkd-dpr}. 

\begin{table}
\footnotesize
\caption{Computational performance of computing the quantum key rates for the dprBB84 protocol with $c$ global phases\diff{, which involves quantum relative entropies with $m\times m$ matrix arguments}. We compare computing this quantity by using the full quantum relative entropy cone (QRE), against using a tailored cone and barrier (QKD), and where we additionally exploit the Hessian being a low-rank perturbation of the log determinant Hessian (DPR).}
\label{tab:qkd-dpr}
\centering
\sisetup{add-integer-zero=false,add-decimal-zero=false}
\begin{tabular*}{\textwidth}{@{\extracolsep{\fill}}rc*{3}{S[table-format=-2.1]}*{3}{S[table-format=-4.2]}*{3}{S[table-format=-1.3]}@{\extracolsep{\fill}}}
\toprule
 & & \multicolumn{3}{l}{\textbf{iter}} & \multicolumn{3}{l}{\textbf{solve time} (s)} & \multicolumn{3}{l}{\textbf{time per iter} (s/iter)} \\
 \cmidrule{3-5}\cmidrule{6-8}\cmidrule{9-11}
$c$ & \diff{$m$} & \multicolumn{1}{c}{DPR} & \multicolumn{1}{c}{QKD} & \multicolumn{1}{c}{QRE} & \multicolumn{1}{c}{DPR} & \multicolumn{1}{c}{QKD} & \multicolumn{1}{c}{QRE} & \multicolumn{1}{c}{DPR} & \multicolumn{1}{c}{QKD} & \multicolumn{1}{c}{QRE} \\ \midrule
2 & \diff{96} & 33 & 32 & 24 & .7 & 1.1 & 2.7 & .02 & .04 & .11 \\
4 & \diff{192} & 27 & 42 & 31 & 5.7 & 21.5 & 24.6 & .21 & .51 & .79 \\
6 & \diff{288} & 26 & 27 & 27 & 27.0 & 64.2 & 129.5 & 1.04 & 2.38 & 4.80 \\
8 & \diff{384} & 28 & 28 & 27 & 77.2 & 311.2 & 697.3 & 2.76 & 11.11 & 25.82 \\
10 & \diff{480} & 28 & 29 & 27 & 278.5 & 1198.2 & 1833.4 & 9.94 & 41.32 & 67.91 \\
\bottomrule
\end{tabular*}
\end{table}

\subsubsection{Discrete-modulated continuous-variable protocol}
\diff{The second protocol we study is the discrete-modulated continuous-variable (DMCV) quantum key distribution protocol~\cite{lin2019asymptotic}. This protocol is parameterized by a cutoff photon number $N_c$, in which case $\mathcal{G}:\mathbb{H}^{4(N_c+1)}\rightarrow\mathbb{H}^{16(N_c+1)}$ is defined by a single Kraus operator}
\begin{equation}\label{eqn:dmcv-kraus}
    K = \sum_{i=1}^4 e_i \otimes \mathbb{I}_4 \otimes P_i,
\end{equation}
\diff{where $e_i\in\mathbb{C}^4$ are the standard $4$-dimensional basis vectors and $P_i\in\mathbb{H}^{N_c+1}_+$ are positive definite measurement operators for $i=1,2,3,4$. The pinching map $\mathcal{Z}:\mathbb{H}^{16(N_c+1)}\rightarrow\mathbb{H}^{16(N_c+1)}$ is defined to act on a $4\times4$ block matrix, i.e., $r=4$. See~\cite[Protocol 2]{lin2019asymptotic} for additional details of this protocol.}

For this protocol, there are no obvious additional structures that we can exploit like the dprBB84 protocol. Therefore, we propose to model the problem using a suitably parameterized cone from Corollary~\ref{cor:main}, as shown in~\eqref{eqn:dprbb84-dpr}, and solve linear systems with the Hessian of the barrier by constructing and Cholesky factoring the Hessian matrix~\eqref{eqn:qkd-hessian}. Consistent with our notation for the dprBB84 experiments, we denote this method as QKD.

Note that as $\mathcal{G}$ consists of a single tall Kraus operator given by~\eqref{eqn:dmcv-kraus}, that $\mathcal{G}(X) \in \mathbb{H}^{16(N_c+1)}$ must be singular with a rank of at most $4(N_c+1)$ for any $X\in\mathbb{H}^{4(N_c+1)}$. Therefore, facial reduction on $\mathcal{G}(X)$ should always be able to reduce from a matrix of size $16(N_c+1)\times16(N_c+1)$ to a matrix of size $4(N_c+1)\times4(N_c+1)$. Similarly, $\mathcal{Z}(\mathcal{G}(X))$ decomposes into four blocks of size $4(N_c+1)\times4(N_c+1)$.

\paragraph{Experimental setup}
We run experiments for DMCV protocols for various cutoff photon numbers $N_c$, and simulate the experimental constraint data $\mathcal{A}$ and $b$ using a source intensity of $\alpha=0.35$, transmission channel length of $L=\SI{60}{\meter}$, and excess noise of $\xi=0.05$. All other experimental parameters are the same as those used in~\cite{hu2022robust}. Like the dprBB84 protocol, we compare QKD against modelling the problem using the full quantum relative entropy cone (QRE), as shown in~\eqref{eqn:dprbb84-qre}. Our results are presented in Table~\ref{tab:qkd-dmcv}. 

\paragraph{Discussion}
From the results, we see that the QRE method performs significantly worse compared to the other methods which use a more tailored cone and barrier. This is due to the facial reduction technique in~\cite[Lemma 6.2]{karimi2023efficient} failing to simplify this problem, and therefore the QRE method must handle matrices of much larger dimensions compared to QKD. Notably, the size of the Hessian matrix of the quantum relative entropy cone used in QRE is $256$ times larger than the size of the Hessian matrices of the tailored cones used by QKD. Additionally, by taking advantage of the block diagonal structure in the linear maps, the QKD method constructs the Hessian matrix by performing matrix multiplications between matrices which have a size $16$ times smaller than those required by the QRE method.
Like for the dprBB84 experiments, we compare our results to those from~\cite{hu2022robust}, which presents results for problems of up to $N_c=11$, compared to our experiments which go up to $N_c=16$. Additionally, \cite{hu2022robust} reports solve times of at least $\SI{500}{\second}$ and $\SI{700}{\second}$ for problems of dimension $N_c=10$ and $11$, respectively, which is slower than our reported solve times for the same problem dimensions.

\subsubsection{\diff{Comparison to prior works}}
\diff{Many works have studied the problem of numerically computing the quantum key rate. In~\cite{coles2016numerical,winick2018reliable}, a first-order Frank-Wolfe method was used to compute these key rates, however this method struggled at times to converge to high accuracy solutions. In~\cite{hu2022robust}, interior-point methods (see also~\cite{faybusovich2020self,karimi2023efficient}) were used with facial reduction techniques to compute the key rates in less time and to higher accuracies compared to these first-order methods. Alternatively, methods which construct a hierarchy of semidefinite programs to lower bound the quantum key rates were proposed in~\cite{brown2021device,araujo2023quantum}. The QKD baseline we present corresponds to the technique proposed in~\cite{lorente2024quantum}. In that paper, the authors present benchmarks showing that the QKD baseline outperforms previous approaches based on interior-point methods and semidefinite approximations.}

\begin{table}
\footnotesize
\caption{Computational performance of computing the quantum key rates for the DMCV protocol\diff{, which involves quantum relative entropies with $m\times m$ matrix arguments}. We compare lifting the problem to the full quantum relative entropy cone (QRE) against using a tailored cone and barrier (QKD).}
\label{tab:qkd-dmcv}
\centering
\sisetup{add-integer-zero=false,add-decimal-zero=false}
\begin{tabular*}{0.8\textwidth}{@{\extracolsep{\fill}}*{1}{S[table-format=-1.0]}c*{2}{S[table-format=-2.1]}*{1}{S[table-format=-3.2]}*{1}{S[table-format=-5.2]}*{1}{S[table-format=-1.2]}*{1}{S[table-format=-2.2]}@{\extracolsep{\fill}}}
\toprule
 & & \multicolumn{2}{l}{\textbf{iter}} & \multicolumn{2}{l}{\textbf{solve time} (s)} & \multicolumn{2}{l}{\textbf{time per iter} (s/iter)} \\
 \cmidrule{3-4}\cmidrule{5-6}\cmidrule{7-8}
$N_c$ & \diff{$m$} & \multicolumn{1}{c}{QKD} & \multicolumn{1}{c}{QRE} & \multicolumn{1}{c}{QKD} & \multicolumn{1}{c}{QRE} & \multicolumn{1}{c}{QKD} & \multicolumn{1}{c}{QRE} \\ \midrule
4 & \diff{80} & 29 & 33 & 3.1 & 1476.7 & 0.1 & 44.7 \\
8 & \diff{144} & 30 & 30 & 10.5 & 20364.1 & 0.4 & 678.8 \\
12 & \diff{208} & 34 & OOM & 84.2 & OOM & 2.5 & OOM \\
16 & \diff{272} & 30 & OOM & 294.7 & OOM & 9.8 & OOM \\
20 & \diff{336} & 33 & OOM & 1029.4 & OOM & 31.2 & OOM \\ \bottomrule
\end{tabular*}
\end{table}

\subsection{Quantum rate-distortion}
\diff{When we transmit information through a communication channel, we are often interested in how to encode the information as efficiently as possible. If we can tolerate some error or distortion in the signal, then we can perform a lossy compression scheme to encode our information. The quantum rate-distortion function~\cite{datta2012quantum,wilde2013quantum} quantifies how much a given quantum information source, represented as a positive definite Hermitian matrix $W\in\mathbb{H}^n_+$ with unit trace, can be compressed without exceeding a maximum allowable distortion $D\geq0$. This distortion is measured as a linear function $X\mapsto\inp{\Delta}{X}$, where $\Delta\in\mathbb{H}^{nm}_+$ is a positive semidefinite matrix. Given this, the \emph{entanglement-assisted rate-distortion function} is given by}
\begin{subequations}\label{eqn:qrd}
    \begin{align}
        R(D) \quad = \quad \minimize_{X \in \mathbb{H}^{nm}} \quad & S\divx{X}{\mathbb{I}\otimes \tr_1^{n,m}(X)} + S(W)  \\
        \subjto \quad & \tr_2^{n,m}(X) = W\\
                \quad & \inp{\Delta}{X} \leq D \\
                \quad & X \succeq 0.
    \end{align}
\end{subequations}
To solve~\eqref{eqn:qrd}, we can directly model the problem using the quantum conditional entropy cone from Corollary~\ref{cor:qce-cone} (see~\eqref{eqn:qrd-qce} for a concrete representation). Additionally, we can efficiently compute inverse Hessian products with the corresponding barrier function as discussed in Example~\ref{exmp:conditional}. We denote this proposed method as QCE.

\paragraph{Experimental setup}
For our experiments, we compute the entanglement-assisted rate-distortion at $D=0.5$ of randomly generated $n\times n$ density matrices sampled uniformly from the Hilbert-Schmidt measure~\cite{zyczkowski2011generating}. We define $\Delta$ to represent the entanglement fidelity distortion~\eqref{eqn:ef-distortion-matrix}. We benchmark using QCE to compute the entanglement-assisted rate-distortion against lifting the problem to the quantum relative entropy cone (QRE), as shown in~\eqref{eqn:qrd-qre}. Our results are summarized in Table~\ref{tab:qrd}.

\paragraph{Discussion}
We see that solving the problem using the QCE method is orders of magnitude faster compared to QRE. This is expected, as using Table~\ref{table:flop-count}, we can predict that solving linear systems with the Hessian of the barrier of the quantum conditional entropy cone should require approximately $n^4$ times fewer flops compared to using the quantum relative entropy cone. Additionally, we observe that fewer iterations are required to converge to the desired tolerance, particularly at larger problem dimensions, which may be attributed to the smaller barrier parameter of the quantum conditional entropy cone.

\begin{table}
\footnotesize
\caption{Computational performance of computing the entanglement-assisted rate-distortion at $D=0.5$ for randomly generated quantum states of size $n\times n$ and using the entanglement fidelity distortion, between conic program formulations using the quantum conditional entropy cone (QCE) and the quantum relative entropy cone (QRE). \diff{Each problem involves quantum relative entropies with $m\times m$ matrix arguments.}}
\label{tab:qrd}
\centering
\sisetup{add-integer-zero=false,add-decimal-zero=false}
\begin{tabular*}{0.85\textwidth}{@{\extracolsep{\fill}}lc*{2}{S[table-format=-2.1]}*{2}{S[table-format=-5.4]}*{2}{S[table-format=-3.5]}@{\extracolsep{\fill}}}
\toprule
 & & \multicolumn{2}{l}{\textbf{iter}} & \multicolumn{2}{l}{\textbf{solve time} (s)} & \multicolumn{2}{l}{\textbf{time per iter} (s/iter)}  \\
 \cmidrule{3-4}\cmidrule{5-6}\cmidrule{7-8}
$n$ & \diff{$m$} & \multicolumn{1}{c}{QCE} & \multicolumn{1}{c}{QRE} & \multicolumn{1}{c}{QCE} & \multicolumn{1}{c}{QRE} & \multicolumn{1}{c}{QCE} & \multicolumn{1}{c}{QRE} \\ \midrule
$2$ & \diff{4} & 10 & 12 & .008 & .048 & .0008 & .0040 \\
$4$ & \diff{16} & 21 & 27 & .047 & 1.746 & .0022 & .0647 \\
$6$ & \diff{36} & 25 & 36 & .272 & 12.178 & .0109 & .3383 \\
$8$ & \diff{64} & 30 & 52 & 1.191 & 145.670 & .0397 & 2.8013 \\
$10$ & \diff{100} & 41 & 61 & 5.394 & 2177.533 & .1316 & 35.6973 \\
$12$ & \diff{144} & 46 & 80 & 14.580 & 11597.550 & .3170 & 144.9694 \\
$14$ & \diff{196} & 49 & 91 & 56.232 & 100624.473 & 1.1476 & 1105.7634 \\
$16$ & \diff{256} & 55 & OOM & 101.571 & OOM & 1.8467 & OOM \\ \bottomrule
\end{tabular*}%
\end{table}

\subsubsection{Entanglement fidelity distortion}
\diff{For some choices of $\Delta$, there exist symmetries that we can exploit to reduce the problem dimension of~\eqref{eqn:qrd}~\cite{he2023efficient}. In particular, let us assume that $W=\diag(w_1,\ldots,w_n)$ is diagonal, which we can do without loss of generality due to unitary invariance of the rate-distortion problem~\cite[Appendix B]{he2023efficient}. Then it is standard to measure distortion using the entanglement fidelity distortion measure, which corresponds to taking $m=n$ and}
\begin{equation}\label{eqn:ef-distortion-matrix}
    \Delta = \mathbb{I} - \sum_{i=1}^n\sum_{j=1}^n \sqrt{w_i w_j} \, e_ie_j^\top \otimes e_ie_j^\top.
\end{equation}
\diff{For this choice of $\Delta$, we know from~\cite[Theorem 4.13]{he2023efficient} that a solution to the problem~\eqref{eqn:qrd} must lie in the subspace}
\begin{equation}\label{eqn:qrd-fixed-point-subspace}
    \mathcal{V} = \biggl\{ \sum_{i \neq j}^n y_{ij} e_ie_i^\top \otimes e_je_j^\top + \sum_{ij}^n Z_{ij} e_ie_j^\top \otimes e_ie_j^\top : y_{ij} \in \mathbb{R} \ \forall i\neq j, \ Z \in \mathbb{H}^n \biggr\},
\end{equation}
\diff{which can be interpreted as a block diagonal subspace with $n^2-n$ blocks of size $1\times1$, and one block of size $n\times n$. This allows us to numerically compute the quantum rate-distortion function more efficiently by only searching over the $(2n^2-n)$-dimensional subspace $\mathcal{V}$, rather than the full $n^4$-dimensional vector space $\mathbb{H}^{n^2}$. In We will do this by modeling the optimization problem with the following cone.}

\begin{corollary}\label{cor:cone-qrd}
    \diff{Let us parameterize the subspace~\eqref{eqn:qrd-fixed-point-subspace} using $(y, Z)\in\mathbb{R}^{n^2-n}\times\mathbb{H}^n$ together with the linear map $\mathcal{G}:\mathbb{R}^{n^2-n}\times\mathbb{H}^n\rightarrow\mathbb{H}^{n^2}$ defined by}
    \begin{equation}\label{eqn:qrd-lin-map}
        \mathcal{G}(y, Z) = \sum_{i \neq j}^n y_{ij} e_ie_i^\top \otimes e_je_j^\top + \sum_{ij}^n Z_{ij} e_ie_j^\top \otimes e_ie_j^\top.
    \end{equation}
    \diff{Then the restriction of~\eqref{eqn:qce-barrier} to the subspace~\eqref{eqn:qrd-fixed-point-subspace}, i.e.,}
    \begin{equation}\label{eqn:qrd-barrier}
        (t, y, Z) \mapsto -\log(t - S\divx{\mathcal{G}(y, Z)}{\mathbb{I}\otimes\tr_1^{n,n}(\mathcal{G}(y, Z))}) - \sum_{i\neq j}^n \log(y_{ij}) - \log\det(Z),
    \end{equation}
    \diff{defined on $\mathbb{R}\times\mathbb{R}^{n^2-n}_{++}\times\mathbb{H}^n_{++}$ is an $(n^2+1)$-logarithmically homogeneous self-concordant barrier for the restriction of the quantum conditional entropy cone to~\eqref{eqn:qrd-fixed-point-subspace}, i.e.,}
    \begin{equation}
        \mathcal{K}_{\textnormal{qrd}} \coloneqq \closure \{ (t, y, Z) \in\mathbb{R}\times\mathbb{R}^{n^2-n}_{++}\times\mathbb{H}^n_{++} : t \geq S\divx{\mathcal{G}(y, Z)}{\mathbb{I}\otimes\tr_1^{n,n}(\mathcal{G}(y, Z))} \}.
    \end{equation}
    \diff{Moreover, this barrier is optimal in the sense that any self-concordant barrier for $\mathcal{K}_{\textnormal{qrd}}$ has parameter at least $n^2+1$.}
\end{corollary}
\begin{proof}
    See Appendix~\ref{subsec:cor-proof}. 
\end{proof}

See~\eqref{eqn:efqrd-qrd} for a concrete representation for how we model the optimization problem using this cone. Using this cone, we can efficiently compute inverse Hessian vector products as follows. First, we have that 
\begin{align}
    S\divx{\mathcal{G}(y, Z)}{\mathbb{I}\otimes\tr_1^{n,n}(\mathcal{G}(y, Z))} &= -S(\mathcal{G}(y, Z)) + S(\tr_1^{n,n}(\mathcal{G}(y, Z))) \nonumber\\
    &= -H(y) - S(Z) + H(\hat{\mathcal{G}}(y, Z)), \label{eqn:qre-entanglement-fidelity}
\end{align}
where $\mathcal{G}:\mathbb{R}^{n^2-n}\times\mathbb{H}^n\rightarrow\mathbb{H}^{n^2}$ is given by~\eqref{eqn:qrd-lin-map}, and $\hat{\mathcal{G}}:\mathbb{R}^{n^2-n}\times\mathbb{H}^n\rightarrow\mathbb{R}^n$ is defined as
\begin{equation*}
    \hat{\mathcal{G}}(y, Z) = \diag(\tr_1^{n,n}(\mathcal{G}(y, Z))) = \biggl[Z_{ii} + \sum_{j \neq i}^n y_{ij}\biggr]_{i=1}^n.
\end{equation*}
The first equality uses Lemma~\ref{lem:partial-trace-entropy}, and the second equality uses~\cite[Corollary 4.14]{he2023efficient} and the block diagonal decomposition technique described in Section~\ref{subsec:blk-diag} (see also Remark~\ref{rem:classical}). Given this, we recognize that $\hat{\mathcal{G}}$ maps to a low-dimensional vector space, and is therefore an instance of the low-rank structures discussed in Example~\ref{exmp:low-rank-small}. Concretely, consider the function
\begin{equation}\label{eqn:f-hat-qrd-ef}
    \hat{F}(y, Z) = \zeta^{-1}S\divx{\mathcal{G}(y, Z)}{\mathbb{I}\otimes\tr_1^{n,n}(\mathcal{G}(y, Z))} - \sum_{i\neq j}^n \log(y_{ij}) - \log\det(Z),
\end{equation}
which, as we recall from Appendix~\ref{sec:barrier-hessian}, we are interested in because solving linear systems with the Hessian of the barrier function~\eqref{eqn:qrd-barrier} can be reduced to solving linear systems with the Hessian of~\eqref{eqn:f-hat-qrd-ef}. Using~\eqref{eqn:qre-entanglement-fidelity}, we can write this Hessian as
\begin{align}\label{eqn:qrd-ef-hessian}
    \grad^2\hat{F}(y, Z) &= \begin{bmatrix}
        \diag( \zeta y )^{-1} + \diag(y)^{-2} & 0 \\ 0 & -\grad^2 (\zeta^{-1} S + \log\det)(Z)
    \end{bmatrix} - \hat{\bm{\mathcal{G}}}^\dag \diag(\hat{\mathcal{G}}(y, Z) )^{-1}  \hat{\bm{\mathcal{G}}}.
\end{align}
Then, with a similar flavor as the quantum conditional entropy from Example~\ref{exmp:conditional}, we have a rank-$n$ perturbation of a $(2n^2-n)\times(2n^2-n)$ dimensional Hessian matrix which we can easily solve linear systems with. Therefore, we can apply the symmetric variant of the matrix inversion lemma~\eqref{eqn:matrix-inversion-lemma-b} to efficiently solve linear systems with the Hessian matrix of the desired barrier. Overall, the complexity of forming and Cholesky factoring the required Schur complement matrix to solve linear systems with~\eqref{eqn:qrd-ef-hessian} is approximately $O(n^3)$ flops. We denote this method for computing the entanglement-assisted rate-distortion as QRD.

\paragraph{Experimental setup}
To illustrate the benefits of taking advantage of structure in this way, we run experiments using the same problem setups as our prior quantum rate-distortion experiments. We benchmark the method QRD against two other methods of utilizing knowledge of the fixed point subspace~\eqref{eqn:qrd-fixed-point-subspace}. First, we consider modelling the problem as~\eqref{eqn:efqrd-qce}, which is similar to QCE but we constrain ourselves to work within the image of the subspace~\eqref{eqn:qrd-fixed-point-subspace} (QCE*). Second, using a similar block diagonal decomposition technique as described in Section~\ref{subsec:blk-diag} and Remarks~\ref{rem:same-block} and~\ref{rem:classical}, we can show that
\begin{align}\label{eqn:qrd-ef-double-re}
    S\divx{\mathcal{G}(y, Z)}{\mathbb{I}\otimes\tr_1^{n,n}(\mathcal{G}(y, Z))} &= H\divx{y}{\hat{\mathcal{G}}_1(y, Z)} + S\divx{Z}{\hat{\mathcal{G}}_2(y, Z)},
\end{align}
for suitable linear maps $\hat{\mathcal{G}}_1:\mathbb{R}^{n^2-n}\times\mathbb{H}^n\rightarrow\mathbb{R}^{n^2-n}$ and $\hat{\mathcal{G}}_2:\mathbb{R}^{n^2-n}\times\mathbb{H}^n\rightarrow\mathbb{H}^{n}$, and therefore we can model the problem as~\eqref{eqn:efqrd-qre} using just the classical and quantum relative entropy cones (QRE*). Our results are summarized in Table~\ref{tab:qrd-ef}.

\paragraph{Discussion}
Comparing to the results from Table~\ref{tab:qrd}, we see that all methods which take advantage of knowledge of the fixed-point subspace outperform the methods which do not. We note that the QCE* results are not significantly different from the QCE results, which is due to the bottleneck in the interior-point methods being in constructing and factoring the Hessian of the cone barrier function, rather than in building and factoring the Schur complement matrix, which is where the main advantage of QCE* over QCE lies. The other two methods, QRD and QRE*, which take advantage of the way the quantum relative entropy function can be decomposed as a sum of simpler functions, show a more significant improvement in computational performance, and the QRD method, which takes advantage of the most structure in the problem, performs the best. We remark that compared to the most basic QRE modelling strategy from Table~\ref{tab:qrd}, we have improved the theoretical flop count complexity of building and factoring the Hessian matrix by a factor of $n^9$, which is reflected in the significantly larger problem instances that we are able to solve.


\begin{table}
\footnotesize
\caption{Computational performance of computing the entanglement-assisted rate-distortion at $D=0.5$ for randomly generated quantum states of size $n\times n$ and using the entanglement fidelity distortion\diff{, which involves quantum relative entropies with $m\times m$ matrix arguments}. We compare between different ways we can account for the fixed-point subspace~\eqref{eqn:qrd-fixed-point-subspace}. These include using a specialized quantum rate-distortion cone for entanglement fidelity distortion (QRD), a simplified model using the classical and quantum relative entropy cones (QRE*), and using the quantum conditional entropy cone while encoding the fixed point subspace using linear constraints (QCE*).}
\label{tab:qrd-ef}
\centering
\sisetup{add-integer-zero=false,add-decimal-zero=false}
\begin{tabular*}{\textwidth}{@{\extracolsep{\fill}}lc*{3}{S[table-format=-2.1]}*{3}{S[table-format=-1.4]}*{3}{S[table-format=-1.5]}@{\extracolsep{\fill}}}
\toprule
 & & \multicolumn{3}{l}{\textbf{iter}} & \multicolumn{3}{l}{\textbf{solve time} (s)} & \multicolumn{3}{l}{\textbf{time per iter} (s/iter)}\\
 \cmidrule{3-5}\cmidrule{6-8}\cmidrule{9-11}
$n$ & \diff{$m$} & \multicolumn{1}{c}{QRD} & \multicolumn{1}{c}{QRE*} & \multicolumn{1}{c}{QCE*} & \multicolumn{1}{c}{QRD} & \multicolumn{1}{c}{QRE*} & \multicolumn{1}{c}{QCE*} & \multicolumn{1}{c}{QRD} & \multicolumn{1}{c}{QRE*} & \multicolumn{1}{c}{QCE*} \\ \midrule
2 & \diff{4} & 10 & 11 & 10 & 0.008 & 0.010 & 0.007 & 0.0007 & 0.0009 & 0.0007 \\
4 & \diff{16} & 20 & 23 & 21 & 0.020 & 0.080 & 0.049 & 0.0010 & 0.0035 & 0.0023 \\
8 & \diff{64} & 32 & 38 & 31 & 0.044 & 0.377 & 1.189 & 0.0014 & 0.0093 & 0.0384 \\
16 & \diff{256} & 53 & 87 & 55 & 0.143 & 3.031 & 119.114 & 0.0027 & 0.0348 & 2.1657 \\
32 & \diff{1024} & 92 & 186 & OOM & 1.414 & 52.954 & OOM & 0.0154 & 0.2847 & OOM \\
64 & \diff{4096} & 137 & OOM & OOM & 11.157 & OOM & OOM & 0.0814 & OOM & OOM \\
\bottomrule
\end{tabular*}%
\end{table}

\subsubsection{\diff{Comparison to prior works}}
\diff{Previously, first-order methods were proposed in~\cite{he2023bregman,he2023efficient,hayashi2023bregman} to compute the quantum rate-distortion function. Empirically, the methods~\cite{he2023efficient,hayashi2023bregman} were shown to converge at a linear rate, and therefore scale better to larger problem dimensions compared to the interior-point methods proposed in this paper, while still reaching high solution accuracies. However, the linear convergence rates of these first-order methods are highly dependent on the parameter $D$, whereas we found interior-point methods to be essentially independent of it.}

\subsection{Quantum channel capacity}
\diff{A closely related concept to the rate-distortion function is the channel capacity, which characterizes the maximum rate of information that can be transmitted reliably through a noisy quantum channel represented by a completely positive trace preserving (see, e.g.,~\cite[Definitions 4.4.2 and 4.4.3]{wilde2017quantum}) linear map $\mathcal{N}:\mathbb{H}^{n_i}\rightarrow\mathbb{H}^{n_o}$. Quantum channel capacities depend on what resources are available. We introduce two of these settings in the following, then present a combined discussion on the experimental results.}


\subsubsection{Entanglement-assisted capacity}
\diff{First, we recognize that any quantum channel $\mathcal{N}$ can always be associated with an isometry matrix $V\in\mathbb{C}^{n_on_e\times n_i}$ such that $\mathcal{N}(X)=\tr_2^{n_o,n_e}(VXV^\dag)$~\cite[Proposition 2.20]{watrous2018theory}. Given this, the \emph{entanglement-assisted channel capacity}~\cite{bennett2002entanglement} for this quantum channel is given by}
\begin{subequations}\label{eqn:ea}
    \begin{align}
        \maximize_{X\in\mathbb{H}^{n_i}} \quad & -S\divx{V X V^\dag}{\mathbb{I} \otimes \tr_1^{m,p}(V X V^\dag)} + S(\tr_2^{m,p}(V X V^\dag)) \label{eqn:qmi-func}\\
        \subjto \quad & \tr[X] = 1 \\
        & X\succeq 0.
    \end{align}
\end{subequations}
\diff{Note that the objective function in~\eqref{eqn:qmi-func} is known as the \emph{quantum mutual information}, which we can model using the following self-concordant barrier.}

\begin{corollary}\label{cor:qmi-cone}
    \diff{Let $V\in\mathbb{C}^{mp\times n}$ be an isometry matrix. The function}
    \begin{equation}
        (t, X) \mapsto -\log(t - [S\divx{V X V^\dag}{\mathbb{I} \otimes \tr_1^{m,p}(V X V^\dag)} - S(\tr_2^{m,p}(V X V^\dag)) + S(\tr[X])]) - \log\det(X),
    \end{equation}
    \diff{defined on $\mathbb{R}\times\mathbb{H}^n_{++}$ is an $(n+1)$-logarithmically homogeneous self-concordant barrier for the quantum mutual information cone, i.e.,}
    \begin{equation}
        \mathcal{K}_{\textnormal{qmi}} \coloneqq \closure \bigl\{ (t, X) \in \mathbb{R}\times\mathbb{H}^n_{++} : t \geq S\divx{V X V^\dag}{\mathbb{I} \otimes \tr_1^{m,p}(V X V^\dag)} - S(\tr_2^{m,p}(V X V^\dag)) + S(\tr[X]) \bigr\},
    \end{equation}
    \diff{Moreover, this barrier is optimal in the sense that any self-concordant barrier for $\mathcal{K}_{\textnormal{qmi}}$ has parameter at least $n+1$.}
\end{corollary}
\begin{proof}
    See Appendix~\ref{subsec:cor-proof}. 
\end{proof}

See~\eqref{eqn:eacc-qmi} for a concrete representation for how we model the entanglement-assisted channel capacity using this cone.
Using Lemma~\ref{lem:partial-trace-entropy}, and the facial reduction technique from Section~\ref{subsubsec:facial-reduction}, we can show that the (homogenized) quantum mutual information can be expressed as
\begin{multline}
    S\divx{V X V^\dag}{\mathbb{I} \otimes \tr_1^{m,p}(VXV^\dag)} - S(\tr_2^{m,p}(VXV^\dag)) + S(\tr[X]) \\
    = -S(X) + S(\tr_1^{m,p}(VXV^\dag)) - S(\tr_2^{m,p}(VXV^\dag)) + S(\tr[X]).
\end{multline}
Therefore we can compute the Hessian of the barrier of the quantum mutual information cone by using a variation of~\eqref{eqn:diff-entropy-hess}. We denote this proposed method as QMI.

\paragraph{Experimental setup}
For our experiments, we compute the entanglement-assisted channel capacity on randomly generated quantum channels with dimensions $n_i=n_o=n_e=n$, i.e., the input, output, and environment systems are all the same dimension. These channels are generated by uniformly sampling Stinespring operators $V$ on the Hilbert-Schmidt measure~\cite{kukulski2021generating}. We benchmark computing the entanglement-assisted channel capacity using QMI against modelling the problem using the Cartesian product between a quantum conditional entropy cone and a quantum entropy cone (QCE), as shown in~\eqref{eqn:eacc-qce}, and modelling the problem using the Cartesian product between a quantum relative entropy cone and a quantum entropy cone (QRE), as shown in~\eqref{eqn:eacc-qre}. Results are summarized in Table~\ref{tab:eacc}.


\begin{table}
\footnotesize
\caption{Computational performance of computing the entanglement-assisted channel capacity of random quantum channels with input, output and environment dimensions $n$, between conic program formulations using the quantum mutual information cone (QMI), the quantum conditional entropy cone (QCE), and the quantum relative entropy cone (QRE). \diff{Each problem involves quantum relative entropies with $m\times m$ matrix arguments.}}
\label{tab:eacc}
\centering
\sisetup{add-integer-zero=false,add-decimal-zero=false}
\begin{tabular*}{\textwidth}{@{\extracolsep{\fill}}lc*{3}{S[table-format=-2.1]}*{3}{S[table-format=-1.4]}*{3}{S[table-format=-1.5]}@{\extracolsep{\fill}}}
\toprule
 & & \multicolumn{3}{l}{\textbf{iter}}& \multicolumn{3}{l}{\textbf{solve time} (s)} &\multicolumn{3}{l}{\textbf{time per iter} (s/iter)} \\
 \cmidrule{3-5}\cmidrule{6-8}\cmidrule{9-11}
$n$ & \diff{$m$} & \multicolumn{1}{c}{QMI} & \multicolumn{1}{c}{QCE} & \multicolumn{1}{c}{QRE} & \multicolumn{1}{c}{QMI} & \multicolumn{1}{c}{QCE} & \multicolumn{1}{c}{QRE} & \multicolumn{1}{c}{QMI} & \multicolumn{1}{c}{QCE} & \multicolumn{1}{c}{QRE} \\ \midrule
2 & \diff{4} & 7 & 10 & 13 & 0.006 & 0.010 & 0.042 & 0.0008 & 0.0010 & 0.0032 \\
4 & \diff{16} & 8 & 20 & 32 & 0.006 & 0.076 & 1.384 & 0.0007 & 0.0038 & 0.0432 \\
8 & \diff{64} & 9 & 23 & 55 & 0.016 & 0.676 & 73.100 & 0.0018 & 0.0294 & 1.3291 \\
16 & \diff{256} & 9 & 30 & OOM & 0.092 & 30.754 & OOM & 0.0102 & 1.0251 & OOM \\
32 & \diff{1024} & 12 & 36 & OOM & 1.602 & 3338.220 & OOM & 0.1335 & 92.7283 & OOM \\
64 & \diff{4096} & 15 & OOM & OOM & 47.311 & OOM & OOM & 3.1541 & OOM & OOM \\
\bottomrule
\end{tabular*}%
\end{table}

\subsubsection{Quantum-quantum capacity of degradable channels}
\diff{Again, consider a quantum channel $\mathcal{N}$ defined as $\mathcal{N}(X)=\tr_2^{n_o,n_e}(VXV^\dag)$ for some isometry matrix $V\in\mathbb{C}^{n_on_e\times n_i}$. The complementary channel $\mathcal{N}_\textnormal{c}:\mathbb{H}^{n_i}\rightarrow\mathbb{H}^{n_e}$ is defined as $\mathcal{N}_\textnormal{c}(X)=\tr_1^{n_o,n_e}(VXV^\dag)$. A degradable channel is a quantum channel $\mathcal{N}$ such that its complementary channel $\mathcal{N}_\textnormal{c}$ can be expressed as $\mathcal{N}_\textnormal{c}=\Xi\circ\mathcal{N}$ for some completely positive trace preserving linear map $\Xi:\mathbb{H}^{n_o}\rightarrow\mathbb{H}^{n_e}$.}

\diff{The second quantum channel capacity we will consider is the \emph{quantum-quantum channel capacity}~\cite{lloyd1997capacity,shor2002quantum,devetak2005private}. In general, this channel capacity is given by a non-convex optimization problem. However, if the channel $\mathcal{N}$ is degradable, then the channel capacity is given by~\cite{devetak2005capacity}}
\begin{subequations}\label{eqn:qq}
    \begin{align}
        \maximize_{X\in\mathbb{H}^{n_i}} \quad & -S\divx{W\mathcal{N}(X) W^\dag}{\mathbb{I} \otimes \tr_1^{n_e,n_{f}}(W\mathcal{N}(X) W^\dag)} \label{eqn:qci-func}\\
        \subjto \quad & \tr[X] = 1 \\
        &X\succeq 0,
    \end{align}
\end{subequations}
\diff{where $W\in\mathbb{C}^{n_en_{f}\times n_o}$ is the isometry matrix associated with $\Xi$ such that $\Xi(X)=\tr_2^{n_e,n_f}(WXW^\dag)$. Note that the objective function in~\eqref{eqn:qci-func} is a simplification of the \emph{quantum coherent information} when $\mathcal{N}$ is degradable, which we can model using the following self-concordant barrier.}

\begin{corollary}\label{cor:qci-cone}
    \diff{Consider a positive linear map $\mathcal{N}:\mathbb{H}^n\rightarrow\mathbb{H}^m$ and an isometry matrix $W\in\mathbb{C}^{pq\times m}$. Then}
    \begin{equation}
        (t, X) \mapsto -\log(t -S\divx{W\mathcal{N}(X) W^\dag}{\mathbb{I} \otimes \tr_1^{p,q}(W\mathcal{N}(X) W^\dag)}) - \log\det(X),
    \end{equation}
    \diff{defined on $\mathbb{R}\times\mathbb{H}^n_{++}$ is an $(n+1)$-logarithmically homogeneous self-concordant barrier for the quantum coherent information cone for degradable channels, i.e.,}
    \begin{equation}
        \mathcal{K}_{\textnormal{qci}} \coloneqq \closure \bigl\{ (t, X) \in \mathbb{R}\times\mathbb{H}^n_{++} : t \geq S\divx{W\mathcal{N}(X) W^\dag}{\mathbb{I} \otimes \tr_1^{p,q}(W\mathcal{N}(X) W^\dag)} \bigr\}.
    \end{equation}
    \diff{Moreover, this barrier is optimal in the sense that any self-concordant barrier for $\mathcal{K}_{\textnormal{qci}}$ has parameter at least $n+1$.}
\end{corollary}

See~\eqref{eqn:qqcc-qci} for a concrete representation for how we can model the quantum-quantum channel capacity using this cone.
Like for the quantum mutual information, we can show that the quantum coherent information can be represented as
\begin{equation}
    S\divx{W\mathcal{N}(X)W^\dag}{\mathbb{I}\otimes\tr_1^{p,q}(W\mathcal{N}(X)W^\dag)} = S(\mathcal{N}(X)) - S(\tr_1^{p,q}(W\mathcal{N}(X)W^\dag)),
\end{equation}
and therefore the Hessian of the barrier function can be constructed using a suitable application of~\eqref{eqn:diff-entropy-hess}. We denote this proposed method as QCI.

\paragraph{Experimental setup}
We compute the quantum-quantum channel capacities of randomly generated pseudo-diagonal quantum channels~\cite{cubitt2008structure}, which are known to be degradable and whose structure is described in~\cite[Section 5]{cubitt2008structure}. We generate these channels such that $n_i=n_o=n_e=n_f=n$, i.e., the dimensions of the input, output, and all other environment systems are the same. We benchmark solving the quantum-quantum channel capacity using QCI with either modelling the problem using the quantum conditional entropy cone (QCE), as shown in~\eqref{eqn:qqcc-qce}, using the quantum relative entropy (QRE) cone, as shown in~\eqref{eqn:qqcc-qre}. Experimental results are shown in Table~\ref{tab:qqcc}.

\begin{table}
\footnotesize
\caption{Computational performance of computing the quantum-quantum channel capacity of random degradable quantum channels with input, output and environment dimensions $n$, between conic program formulations using the quantum coherent information cone (QCI), the quantum conditional entropy cone (QCE), and the quantum relative entropy cone (QRE). \diff{Each problem involves quantum relative entropies with $m\times m$ matrix arguments.}}
\label{tab:qqcc}
\centering
\sisetup{add-integer-zero=false,add-decimal-zero=false}
\begin{tabular*}{\textwidth}{@{\extracolsep{\fill}}lc*{3}{S[table-format=-2.1]}*{3}{S[table-format=-2.4]}*{3}{S[table-format=-1.5]}@{\extracolsep{\fill}}}
\toprule
 & & \multicolumn{3}{l}{\textbf{iter}} & \multicolumn{3}{l}{\textbf{solve time} (s)}& \multicolumn{3}{l}{\textbf{time per iter} (s/iter)}\\
 \cmidrule{3-5}\cmidrule{6-8}\cmidrule{9-11}
$n$ & \diff{$m$} & \multicolumn{1}{c}{QCI} & \multicolumn{1}{c}{QCE} & \multicolumn{1}{c}{QRE} & \multicolumn{1}{c}{QCI} & \multicolumn{1}{c}{QCE} & \multicolumn{1}{c}{QRE} & \multicolumn{1}{c}{QCI} & \multicolumn{1}{c}{QCE} & \multicolumn{1}{c}{QRE} \\ \midrule
2 & \diff{4} & 7 & 13 & 14 & 0.006 & 0.010 & 0.053 & 0.0001 & 0.0001 & 0.0038 \\
4 & \diff{16} & 6 & 20 & 30 & 0.007 & 0.071 & 1.624 & 0.0012 & 0.0035 & 0.0541 \\
8 & \diff{64} & 9 & 25 & 53 & 0.018 & 0.546 & 76.524 & 0.0020 & 0.0219 & 1.4438 \\
16 & \diff{256} & 13 & 51 & OOM & 0.147 & 33.813 & OOM & 0.0113 & 0.8247 & OOM \\
32 & \diff{1024} & 13 & 35 & OOM & 1.324 & 3310.388 & OOM & 0.1018 & 94.5825 & OOM \\
64 & \diff{4096} & 15 & OOM & OOM & 33.927 & OOM & OOM & 2.2618 & OOM & OOM \\
\bottomrule
\end{tabular*}%
\end{table}

\subsubsection{Discussion}
Overall, we see that the results for the entanglement-assisted and quantum-quantum channel capacities are very similar, due to the two problems sharing very similar structures. We first observe that the QCE methods perform significantly better than the na\"ive QRE methods in both computation time and memory requirements. Referring to Table~\ref{table:flop-count}, we see that the QRE methods require $O(n^{12})$ flops to perform an inverse Hessian product, whereas QCE methods only require $O(n^8)$ flops. We also observe that the tailored QMI and QCI are able to provide a further order of magnitude improvement in computation times and memory requirements, which only require $O(n^6)$ flops to compute inverse Hessian products. Additionally, we observe that modelling the QCI and QMI methods require significantly  fewer iterations than the other methods. We attribute this to the smaller barrier parameter of the quantum mutual information and quantum coherent information cones ($1+n$), compared to when we use the quantum conditional entropy ($1+n^2$) and quantum relative entropy ($1+2n^2$) to model the problem. 


\subsubsection{\diff{Comparison to prior works}}

\diff{Previously, interior-point methods were used to compute these capacities in~\cite{fawzi2018efficient,coey2023performance}. However, these approaches, correspoding to our QRE baselines, were limited to computing the capacity for a single-qubit amplitude damping channel (i.e., where $n_i=n_o=n_e=n_f=2$).
In~\cite{ramakrishnan2020computing,he2023bregman}, first-order methods were proposed to numerically compute the entanglement-assisted channel capacities of channels of comparable sizes to those we explore in this paper. To our knowledge, there is no first-order method which has been proposed to compute the quantum-quantum channel capacity of degradable channels.}

\subsection{Ground state energy of Hamiltonians}
\diff{A fundamental question in quantum many-body theory is to compute the ground energy of local Hamiltonians~\cite{gharibian2015quantum}. A translation-invariant 2-local Hamiltonian on an infinite one-dimensional chain is described by a single $4\times 4$ Hermitian matrix $h$. In~\cite{fawzi2023entropy}, it was shown that the ground energy density of such a Hamiltonian can be lower bounded by solving the convex optimization problem}
\begin{subequations}\label{eqn:ccm}
    \begin{align}
        \minimize_{X\in\mathbb{H}^{2^l}} \quad & \inp{h\otimes \mathbb{I}_{2^{l-2}}}{X}  \\
        \subjto \quad & \tr_1^{2,2^{l-1}}(X) = \tr_2^{2^{l-1},2}(X) \\
        & \tr[X] = 1 \\
        & S\divx{X}{\mathbb{I}\otimes\tr_1^{2,2^{l-1}}(X)} \leq 0 \\
        & X \succeq 0,
    \end{align}
\end{subequations}
\diff{where $l$ controls the accuracy of the lower bound which converges to the true value as $l\rightarrow\infty$.}

To compute this quantity, we directly optimize over the quantum conditional entropy cone from Corollary~\ref{cor:qce-cone} (see~\eqref{eqn:med-qce} for a concrete representation), and efficiently solve linear systems with the Hessian matrix of the corresponding barrier using the method discussed in Example~\ref{exmp:conditional}. We denote this proposed method as QCE.


\paragraph{Experimental setup}
We are interested in computing a hierarchy of lower bounds for the ground state energy of the XXZ-Hamiltonian parameterized by $l$, i.e., where
\begin{equation}
    H = (-\sigma_x \otimes \sigma_x - \sigma_y \otimes \sigma_y - \Delta \sigma_z \otimes \sigma_z) \otimes \mathbb{I}_{2^{l-2}},
\end{equation}
where $\sigma_x,\sigma_y,\sigma_z\in\mathbb{H}^2$ are the Pauli matrices
\begin{equation*}
    \sigma_x = \begin{bmatrix}
        0 & 1 \\ 1 & 0
    \end{bmatrix}, \qquad \sigma_y = \begin{bmatrix}
        0 & -i \\ i & 0
    \end{bmatrix}, \qquad \sigma_z = \begin{bmatrix}
        1 & 0 \\ 0 & -1
    \end{bmatrix}.
\end{equation*}
We consider the case where $\Delta=-1$. We benchmark estimating the ground state energy for various values of $l$ using QCE against lifting the problem to the quantum relative entropy cone (QRE), as shown in~\eqref{eqn:med-qre}. Experimental results are presented in Table~\ref{tab:geh}.

\paragraph{Discussion}
We begin by remarking that for this problem, the quantum conditional entropy traces out a $2$-dimensional system, which corresponds to the smallest non-trivial dimension we can trace out. From Table~\ref{table:flop-count}, we see that this corresponds to the scenario when exploiting the structure of the barrier of the quantum conditional entropy cone yields the least computational benefits. Nevertheless, we still expect some benefits, e.g., the Cholesky facotrization step for the quantum conditional entropy barrier should use approximately $64$ times fewer flops compared to the quantum relative entropy barrier. This is reflected in the results, where QCE yields between a $30$ to $80$ times speedup compared to QRE. Additionally, the QCE is more memory efficient, whereas there is insufficient memory to run QRE at the largest problem dimension $l=8$. The number of iterations between the two methods is mostly the same, however we note that at $l=6$ and $7$, the Hypaita algorithm when solving QRE was observed to fall back to less efficient stepping methods at times, which resulted in the increased number of iterations. 




\begin{table}
\footnotesize
\caption{Computational performance of computing a hierarchy of lower bounds for the ground state energies of Hamiltonians, parameterized by the number of qubit systems $l$, between conic program formulations using the quantum conditional entropy cone (QCE) and the quantum relative entropy cone (QRE). We note that the QRE experiment for $l=2$ terminated early before reaching the desired tolerance, and only reached a relative gap of $3.6\times10^{-8}$. \diff{Each problem involves quantum relative entropies with $m\times m$ matrix arguments.}}
\label{tab:geh}
\centering
\sisetup{add-integer-zero=false,add-decimal-zero=false}
\begin{tabular*}{0.9\textwidth}{@{\extracolsep{\fill}}lc*{2}{S[table-format=-2.1]}*{2}{S[table-format=-5.4]}*{2}{S[table-format=-3.5]}@{\extracolsep{\fill}}}
\toprule
 & & \multicolumn{2}{l}{\textbf{iter}} & \multicolumn{2}{l}{\textbf{solve time} (s)} & \multicolumn{2}{l}{\textbf{time per iter} (s/iter)} \\
 \cmidrule{3-4}\cmidrule{5-6}\cmidrule{7-8}
$l$ & \diff{$m$} & \multicolumn{1}{c}{QCE} & \multicolumn{1}{c}{QRE} & \multicolumn{1}{c}{QCE} & \multicolumn{1}{c}{QRE} & \multicolumn{1}{c}{QCE} & \multicolumn{1}{c}{QRE} \\ \midrule
$2$ & \diff{4} & 8 & 11 & .007 & .107 & .0009 & .0097 \\
$3$ & \diff{8} & 9 & 11 & .012 & .150 & .0013 & .0136 \\
$4$ & \diff{16} & 20 & 23 & .065 & 3.420 & .0033 & .1487 \\
$5$ & \diff{32} & 28 & 31 & .541 & 14.841 & .0193 & .4787 \\
$6$ & \diff{64} & 32 & 60 & 6.106 & 525.990 & .1908 & 8.7665 \\
$7$ & \diff{128} & 35 & 98 & 193.430 & 26050.487 & 5.5266 & 265.8213 \\
$8$ & \diff{256} & 39 & OOM & 8288.689 & OOM & 212.5305 & OOM \\ \bottomrule
\end{tabular*}%
\end{table}

\subsubsection{\diff{Comparison to prior works}}
\diff{First-order splitting methods were previously used in~\cite{fawzi2023entropy} to solve for the relaxation~\eqref{eqn:ccm} for problem sizes of up to $l=8$, the same dimensions our proposed interior-point method can solve problems up to. The fact that we can use interior-point methods to solve problems of the same dimension as first-order methods is quite remarkable as first-order methods generally scale to much larger problems than interior-point methods. Furthermore, first-order methods usually require manual tuning of parameters to obtain reliable results, unlike interior-point methods which require little or no tuning.
}

\section{Concluding remarks}


Overall, we have shown how by using cones and barriers tailored to specific problems arising in quantum information theory, we are able to solve these problems using interior-point methods much more efficiently than if we lifted them to the full quantum relative entropy cone. In particular, we showed how we can absorb positive linear maps into the quantum relative entropy cone, and established optimal self-concordance of the natural barriers for these cones. Additionally, we showed how we can exploit common structures in these linear maps, including as block diagonal, low-rank, and quantum conditional entropy-like structures, to more efficiently solve linear systems with the Hessian matrices of these barriers, which is the main bottleneck in interior-point methods for quantum relative entropy \diff{optimization problems}. Our numerical experiments demonstrate that we can solve problems in fewer iterations, with faster per iteration cost, and requiring less computational memory, allowing us to solve problems which were previously intractable, such as to compute quantum key rates of high-dimensional quantum key distribution protocols. We conclude by presenting two directions for future work.

\paragraph{Automated software}
One of the main difficulties in implementing our techniques is that each problem requires its own tailored cone and barrier oracle which users must implement by themselves. As we have generalized the techniques that we have been using in Section~\ref{sec:structure}, it would be beneficial to implement software which is able to automatically detect and exploit these structures in a way that mirrors how state-of-the-art semidefinite programming software is able to automatically take advantage of useful problem structure. 

\paragraph{Primal-dual scalings}
Most of this paper has focused on the efficient practical implementation of quantum relative entropy \diff{optimization problems} for interior-point methods. However, it is also interesting if there are any algorithmic improvements we can make to interior-point methods for solving quantum relative entropy \diff{optimization problems}. One line of work in nonsymmetric cone programming aims to use primal-dual scalings to allow for symmetric primal and dual steps to be taken~\cite{tunccel2001generalization,myklebust2014interior,myklebust2015primal}. This is a desirable property for these algorithms, and is commonly used in symmetric cone programming. This has recently been implemented for the exponential cone in~\cite{dahl2022primal}, and therefore it is of interest whether we can use the same technique to improve methods for solving quantum relative entropy \diff{optimization problems}. However, the primal-dual scaling matrices for non-symmetric cones require oracles for the conjugate barrier. In~\cite{kapelevich2022computing}, it was shown how the conjugate barrier could be efficiently computed for a range of nonsymmetric cones, and in~\cite{karimi2024domain} it was shown how the conjugate barrier for the quantum entropy cone could be computed. However, whether the conjugate barrier for the quantum relative entropy cone can be efficiently computed is still an open question.



{\appendix
\numberwithin{equation}{section}

\section{Proofs of Corollaries~\ref{cor:cone-qrd} and~\ref{cor:qmi-cone}} \label{subsec:cor-proof}
Here, we show how slight modifications to the proof for Theorem~\ref{thm:main} can yield useful barriers for cones closely related to~\eqref{eqn:main-set}. The first result uses the fact that a function which is compatible with a given domain is also compatible with any closed convex subset of this domain.
\begin{proof}[Proof of Corollary~\ref{cor:cone-qrd}]
    From Corollary~\ref{cor:main-compatibility}, we know by choosing $g(x)=\log(x)$, $\mathcal{N}_1(X)=X$, and $\mathcal{N}_2(X)=\mathbb{I}\otimes \tr_1^{n,n}(X)$ that negative quantum conditional entropy 
    \begin{equation*}
        X\mapsto -S\divx{X}{\mathbb{I}\otimes \tr_1^{n,n}(X)},
    \end{equation*}
    defined on $\mathbb{H}^{n^2}_{++}$ is $1$-compatible with respect to the domain $\mathbb{H}^{n^2}_+$. From Definition~\ref{defn:compatibility}, it is clear that if a function is $\beta$-compatible with respect to a domain, then the restriction of this function to any (relatively) open convex subset of its domain is also $\beta$-compatible with respect to the closure of this restricted domain. Therefore, the restriction of negative quantum conditional entropy to the domain $\mathbb{H}^{n^2}_{++}\cap\mathcal{V}$, where $\mathcal{V}$ is given by~\eqref{eqn:qrd-fixed-point-subspace}, i.e.,
    \begin{equation*}
        (y,Z)\mapsto -S\divx{\mathcal{G}(y, Z)}{\mathbb{I}\otimes\tr_1^{n,n}[\mathcal{G}(y, Z)]},
    \end{equation*}
    defined on $\mathbb{R}^{n^2-n}_{++}\times\mathbb{H}^n_{++}$ is $1$-compatible with respect to $\mathbb{H}^{n^2}_{+}\cap\mathcal{V}$, i.e., $(y,Z)\in\mathbb{R}^{n^2-n}_+\times\mathbb{H}^n_+$. The remainder of the proof follows from an identical argument as the proof for Theorem~\ref{thm:main}. 
\end{proof}
The next result uses the fact that the sum of two functions which are compatible with the same domain will also be compatible with this domain.
\begin{proof}[Proof of Corollary~\ref{cor:qmi-cone}]
    From Corollary~\ref{cor:main-compatibility}, we know by choosing $g(x)=\log(x)$, $\mathcal{N}_1(X)=VXV^\dag$ and $\mathcal{N}_2(X)=\mathbb{I}\otimes\tr_1^{m,p}(VXV^\dag)$ that 
    \begin{equation*}
        X \mapsto -S\divx{VXV^\dag}{\mathbb{I}\otimes\tr_1^{m,p}(VXV^\dag)},
    \end{equation*}
    defined on $\mathbb{H}^n_{++}$ is $1$-compatible with respect to $\mathbb{H}^n_+$. Similarly, by choosing $\mathcal{N}_1(X)=\tr_2^{m,p}(VXV^\dag)$ and $\mathcal{N}_2(X)=\tr[X]\mathbb{I}$, we know that
    \begin{align*}
        X &\mapsto -S\divx{\tr_2^{m,p}(VXV^\dag)}{\tr[X]\mathbb{I}} \\
        &=  S(\tr_2^{m,p}(VXV^\dag)) - S(\tr[X]),
    \end{align*}
    defined on $\mathbb{H}^n_{++}$ is also $1$-compatible with respect to $\mathbb{H}^n_+$, where we have used the fact that $X\mapsto \tr_2^{m,p}(VXV^\dag)$ is a trace-preserving map. From Definition~\ref{defn:compatibility}, it is clear that the sum of any two functions which are compatible with respect to the same domain is also compatible with respect to the same domain. Therefore, the (homogenized) quantum mutual information
    \begin{equation*}
        X \mapsto -S\divx{VXV^\dag}{\mathbb{I}\otimes\tr_1^{m,p}(VXV^\dag)} + S(\tr_2^{m,p}(VXV^\dag)) - S(\tr[X]),
    \end{equation*}
    defined on $\mathbb{H}^n_{++}$ is also $1$-compatible with respect to $\mathbb{H}^n_+$. The remainder of the proof follows from an identical argument as the proof for Theorem~\ref{thm:main}. 
\end{proof}

\section{Matrix inversion lemma}\label{appdx:matrix-inversion-lemma}
Consider the matrix $A - UBU^\top$ where $A\in\mathbb{R}^{n\times n}$, $B\in\mathbb{R}^{r\times r}$, and $U\in\mathbb{R}^{n\times r}$. When $r\ll n$, we can interpret this as a low rank perturbation of the matrix $A$, in which case it is well known that the inverse can be expressed as a low rank perturbation of $A^{-1}$. More concretely, if $A$ and $A - UBU^\top$ are nonsingular, then
\begin{equation}\label{eqn:matrix-inversion-lemma-a}
    (A - UBU^\top)^{-1} = A^{-1} + A^{-1}U(\mathbb{I} - BU^\top A^{-1}U)^{-1}BU^\top A^{-1}
\end{equation}
Note that this expression does not require $B$ to be nonsingular. This expression is well defined due to the following fact (see, e.g.,~\cite{henderson1981deriving}).
\begin{fact}
    The matrices $A$ and $A - UBU^\top$ are both nonsingular if and only if $\mathbb{I} - BU^\top A^{-1}U$ is also nonsingular.
\end{fact}
\begin{proof}
    This is a consequence of the identity $\det(\mathbb{I} - BU^\top A^{-1}U)=\det(A^{-1})\det(A-UBU^\top)$. 
\end{proof}
When $B$ is also nonsingular, we have the following alternative expression for the matrix inverse, often referred to as the Sherman-Morrison-Woodbury identity
\begin{equation}\label{eqn:matrix-inversion-lemma-b}
    (A - UBU^\top)^{-1} = A^{-1} + A^{-1}U(B^{-1} - U^\top A^{-1}U)^{-1}U^\top A^{-1}
\end{equation}
Similar to the previous identity, this identity is well defined under suitable assumptions.
\begin{fact}
    The matrices $A$, $B$, and $A - UBU^\top$ are all nonsingular if and only if $B^{-1} - U^\top A^{-1}U$ is also nonsingular.
\end{fact}
\begin{proof}
    This follows from $\det(B^{-1} - U^\top A^{-1}U)=\det(A^{-1})\det(B^{-1})\det(A-UBU^\top)$. 
\end{proof}
This expression is advantageous as it gives a symmetric expression for the matrix inverse. If $A$ and $B$ are both positive definite, then we have the additional property, which is a straightforward consequence of the Schur complement lemma.
\begin{fact}
    If $B\succ0$ and $A-UBU^\top\succ0$, then $B^{-1} - U^\top A^{-1} U \succ0$.
\end{fact}
\begin{proof}
    Consider the matrix
    \begin{equation*}
        X = \begin{bmatrix}
            A & U \\ U^\top & B^{-1}
        \end{bmatrix}.
    \end{equation*}
    Using the Schur complement lemma (see, e.g.,~\cite[Section A.5.5]{boyd2004convex}) for the Schur complement of $B^{-1}$ on $X$, we have that $B^{-1}\succ0$ and $A-UBU^\top\succ0$ implies $X\succ0$. Now using the Schur complement lemma for the Schur complement of $A$ on $X$, we have that $X\succ0$ implies that $A\succ0$ and $B^{-1} - U^\top A^{-1} U \succ0$, as desired. 
\end{proof}
Therefore, under these assumptions, we can perform a Cholesky decomposition on the Schur complement matrix $B^{-1} - U^\top A^{-1}U$ rather than an LU decomposition to solve linear systems with this matrix.

We will also consider the following variant of the matrix inversion lemma.
\begin{lemma}\label{lem:matrix-inversion-lemma-c}
    Consider the matrices $A\in\mathbb{R}^{n\times n}$, $B\in\mathbb{R}^{n\times r}$, $C\in\mathbb{R}^{r\times r}$, and $U\in\mathbb{R}^{n\times r}$. Then if $A$ and $A+UB^\top+BU^\top-UCU^\top$ are both positive definite, then
    \begin{equation*}
        \biggl( \begin{bmatrix} \mathbb{I} & U \end{bmatrix} \begin{bmatrix}
            A      & B \\
            B^\top & -C
        \end{bmatrix} \begin{bmatrix} \mathbb{I} \\ U^\top \end{bmatrix} \biggl)^{-1} = A^{-1} + A^{-1}\begin{bmatrix}
            B & U
        \end{bmatrix} S^{-1} \begin{bmatrix}
            B^\top \\ U^\top
        \end{bmatrix}A^{-1},
    \end{equation*}
    where
    \begin{equation*}
        S = \begin{bmatrix}
            C & \mathbb{I} \\ \mathbb{I} & 0
        \end{bmatrix} + \begin{bmatrix}
            B^\top \\ U^\top
        \end{bmatrix} A^{-1} \begin{bmatrix}
            B & U
        \end{bmatrix}.
    \end{equation*}
\end{lemma}
\begin{proof}
    First, we can recast solving linear systems with the desired matrix and a left-hand side vector $d\in\mathbb{R}^n$ as the optimality conditions of the  unconstrained convex quadratic minimization problem
    \begin{equation*}
        \min_{x\in\mathbb{R}^{n}} \quad x^\top \begin{bmatrix} \mathbb{I} & U \end{bmatrix} \begin{bmatrix}
            A      & B \\
            B^\top & -C
        \end{bmatrix} \begin{bmatrix} \mathbb{I} \\ U^\top \end{bmatrix} x - d^\top x,
    \end{equation*}
    which, in turn, is equivalent to the constrained convex quadratic minimization problem
    \begin{equation*}
        \min_{x\in\mathbb{R}^{n}, y\in\mathbb{R}^r} \quad \begin{bmatrix} x^\top & y^\top \end{bmatrix} \begin{bmatrix}
            A      & B \\
            B^\top & -C
        \end{bmatrix} \begin{bmatrix} x \\ y \end{bmatrix} x - d^\top x, \qquad \subjto \quad y = U^\top x.
    \end{equation*}
    The optimality conditions for this reformulated problem are for the optimal primal and dual variables $x\in\mathbb{R}^n$, $y,z\in\mathbb{R}^r$ to satisfy 
    \begin{equation*}
        \begin{bmatrix}
            A      & B & U \\
            B^\top & -C & -\mathbb{I} \\
            U^\top & -\mathbb{I} & 0
        \end{bmatrix} \begin{bmatrix}
            x \\ y \\ z
        \end{bmatrix} = \begin{bmatrix}
            d \\ 0 \\ 0
        \end{bmatrix}.
    \end{equation*}
    As $A+UB^\top+BU^\top-UCU^\top$ is assumed to be nonsingular, this minimization problem has a unique primal and dual solution, and therefore the KKT matrix must be nonsingular (see, e.g.,~\cite[Section 10.1.1]{boyd2004convex}). Performing a suitable block elimination on this system of equations to solve for $x$ gives the desired result. The expression is well defined as the (nonzero) determinant of the KKT matrix must equal $\det(A)\det(S)$ (see, e.g.,~\cite[Appendix A.5.5]{boyd2004convex}), which implies that $S$ must be nonsingular. 
\end{proof}

\section{Derivatives} \label{appdx:derivatives}
\subsection{Spectral functions} \label{sec:derivatives}
Consider a scalar valued function $f:\domain f\rightarrow\mathbb{R}$. Spectral functions extend $f$ to act on Hermitian matrices $X$ with spectral decomposition $X=\sum_{i=1}^n\lambda_i v_iv_i^\dag$ and eigenvalues satisfying $\lambda_i\in\domain f$ for all $i=1,\ldots,n$ by letting
\begin{equation*}
    f(X) \coloneqq \sum_{i=1}^n f(\lambda_i) v_i v_i^\dag.
\end{equation*}
The derivatives of these functions are fairly well-known. To characterize these derivatives, consider a real diagonal matrix $\Lambda=\diag(\lambda_1,\ldots,\lambda_n)$ with entries $\lambda_i\in\domain f$ for all $i=1,\ldots,n$. We define the first divided differences matrix $f^{[1]}(\Lambda)\in\mathbb{S}^n$ of $\Lambda$ for the function $f$ as the matrix whose $(i,j)$-th entry is $f^{[1]}(\lambda_i, \lambda_j)$, where%
\begin{align*}
    f^{[1]}(\lambda, \mu) &= \frac{f(\lambda) - f(\mu)}{\lambda - \mu}, \quad \textrm{if }\lambda\neq\mu\\
    f^{[1]}(\lambda, \lambda) &= f'(\lambda).
\end{align*}
We also define the $k$-th second divided difference matrix $f^{[2]}_k(\Lambda)\in\mathbb{S}^n$ of the matrix $\Lambda$ for the function $f$ as the matrix whose $(i,j)$-th entry is $f^{[2]}(\lambda_i, \lambda_j, \lambda_k)$, where
\begin{align*}
    f^{[2]}(\lambda, \mu, \sigma) &= \frac{f^{[1]}(\lambda, \sigma) - f^{[1]}(\mu, \sigma)}{\lambda - \mu}, \quad \textrm{if }\lambda\neq\mu \\
    f^{[2]}(\lambda, \lambda, \lambda) &= \frac{1}{2}f''(\lambda),
\end{align*}
and all other scenarios are defined by noticing that $f^{[2]}(\lambda, \mu, \sigma)$ is defined symmetrically in each of its arguments. Given these, the following two lemmas follow from~\cite[Theorem 3.23]{hiai2014introduction} and~\cite[Theorem 3.33]{hiai2014introduction}.
\begin{lemma}\label{eqn:lem-tr-f-derivative}
    Consider the function $g(X)=\tr[f(X)]$ for a twice continuously differentiable function $f:\domain f\rightarrow\mathbb{R}$. The first and second derivatives of $g$ are
    \begin{subequations}
        \begin{align}
            \grad g(X) &= f'(X)\\
            \grad^2g(X)[H] &= U[(f')^{[1]}(\Lambda) \odot (U^\dag H U)]U^\dag, \label{eqn:tr-f-second-derivative}
        \end{align}
    \end{subequations}
    where $X$ has spectral decomposition $X=U\Lambda U^\dag$.
\end{lemma}
\begin{lemma}\label{eqn:lem-tr-c-f-derivative}
    Consider the function $h(X)=\tr[Cf(X)]$ for a Hermitian matrix $C\in\mathbb{H}^n$ and a twice continuously differentiable function $f:\domain f\rightarrow\mathbb{R}$. The first and second derivatives of $h$ are
    \begin{subequations}
        \begin{align}
            \grad h(X) &= U[(f')^{[1]}(\Lambda) \odot (U^\dag C U)]U^\dag\\
            \grad^2h(X)[H] &= U \biggl[ \sum_{k=1}^n (f')^{[2]}_k(\Lambda) \odot \biggl([ U^\dag CU]_k[ U^\dag HU]_k^\dag + [U^\dag HU]_k[U^\dag CU]_k^\dag \biggr) \biggr]U^\dag, \label{eqn:tr-c-f-second-derivative}
        \end{align}
    \end{subequations}
    where $X$ has spectral decomposition $X=U\Lambda U^\dag$.
\end{lemma}

\begin{remark}\label{rem:matrix-form-derivative}
    Sometimes, it can be convenient to represent the second derivatives in matrix form, i.e., to characterize the Hessian matrix, so that we can perform standard linear algebra techniques on these linear maps. For a matrix $X\in\mathbb{H}^n$, we define $\vect(X)$ as the operation which stacks the columns of $X$ to form an $n^2$-dimensional vector. Then using the identity $\vect(AXB^\dag)=(\conj{B}\otimes A)\vect(X)$, the Hessian, as the linear operator acting on vectorized matrices, corresponding to~\eqref{eqn:tr-f-second-derivative} is equivalent to
    \begin{equation}\label{eqn:tr-f-hessian}
        \grad^2 g(X) = (\conj{U}\otimes U) \diag(\vect((f')^{[1]}(\Lambda))) (\conj{U}\otimes U)^\dag,
    \end{equation}
    and the Hessian corresponding to~\eqref{eqn:tr-c-f-second-derivative} is equivalent to~\cite{faybusovich2020self}
    \begin{equation}\label{eqn:tr-c-f-hessian}
        \grad^2 h(X) = (\conj{U}\otimes U) S_C(X) (\conj{U}\otimes U)^\dag,
    \end{equation}
    where $S_C(X)$ is a sparse matrix given by
    \begin{equation*}
        [S_C(X)]_{ij,kl} = \delta_{kl} [U^\dag C U]_{ij} (f')^{[2]}(\lambda_i, \lambda_j, \lambda_l) + \delta_{ij} [U^\dag C U]_{kl} (f')^{[2]}(\lambda_j, \lambda_k, \lambda_l).
    \end{equation*}
    \diff{Constructing the Hessian matrices using the expressions~\eqref{eqn:tr-f-hessian} and~\eqref{eqn:tr-c-f-hessian} costs $O(n^6)$ flops (due to matrix multiplication between $n^2\times n^2$ matrices). In practice, we instead construct the Hessian matrices by building each column of the matrix by applying~\eqref{eqn:tr-f-second-derivative} and~\eqref{eqn:tr-c-f-second-derivative} to the standard basis, which only costs $O(n^5)$ flops.}
\end{remark}

\begin{remark}\label{rem:easy-inverse-hessian}
    The inverse of the second derivative map~\eqref{eqn:tr-f-second-derivative} is given by
    \begin{equation*}
        (\grad^{2}g(X))^{-1}[H] = U[(U^\dag H U) \oslash (f')^{[1]}(\Lambda)]U^\dag,
    \end{equation*}
    where $\oslash$ denotes the elementwise or Hadamard division, and is therefore relatively easy to apply. Although the second derivative map~\eqref{eqn:tr-c-f-second-derivative} is also highly structured, it is not obvious how we can similarly compute its inverse efficiently without building and factoring the Hessian matrix.
\end{remark}

As an important example, consider the case when $f(x)=\log(x)$ defined on $\mathbb{R}_{++}$, in which case we recover the log determinant function $\log\det(X)=\tr[\log(x)]$ defined on $\mathbb{H}^n_{++}$. Using the fact that the first divided differences of $f'(x)=1/x$ is
\begin{equation*}
    (f')^{[1]}(\Lambda) = -[1 / \lambda_i\lambda_j]_{i,j=1}^n,
\end{equation*}
with Lemma~\ref{eqn:lem-tr-f-derivative}, we recover the well-known result
\begin{align}
    \grad^2\log\det(X)[H] &= U[(f')^{[1]}(\Lambda) \odot (U^\dag H U)]U^\dag \nonumber \\
                        &= -U[\Lambda^{-1} (U^\dag H U) \Lambda^{-1}]U^\dag \nonumber \\
                        &= -X^{-1} H X^{-1}.
\end{align}
In the matrix form discussed in Remark~\ref{rem:matrix-form-derivative}, we can represent this as $\grad^2 \log\det(X) = -\conj{X}^{-1} \otimes X^{-1}$. It is also fairly straightforward to show that the inverse of this linear map is
\begin{align}\label{eqn:logdet-hess-inv}
    (\grad^2\log\det(X))^{-1}[H] &= -X H X.
\end{align}

\subsection{Barrier functions}\label{sec:barrier-hessian}
We are interested in barrier functions $F:\mathbb{R}\times\mathbb{H}^n_{++}\rightarrow\mathbb{R}$ of the form
\begin{equation}
    F(t, X) = -\log(t - \varphi(X)) - \log\det(X),
\end{equation}
for some convex function $\varphi:\mathbb{H}^n_{++}\rightarrow\mathbb{R}$, which turns out to be a suitable LHSCB for many cones defined as the epigraphs of spectral functions and related functions~\cite{coey2023conic,fawzi2023optimal,faybusovich2017matrix}. A straightforward computation shows that the first derivatives can be represented as
\begin{align*}
    \grad_t F(t, X) &= -\zeta^{-1}. \\
    \grad_X F(t, X) &= \zeta^{-1} \grad\varphi(X) - X^{-1},
\end{align*}
where $\zeta=t-\varphi(X)$, and, with some abuse of notation, the Hessian, as a linear operator acting on vectorized matrices, can be represented as the block matrix
\begin{align}
    \grad^2 F(t, X) &= \begin{bmatrix}
        \zeta^{-2} & -\zeta^{-2}\grad\varphi(X)^\top \\
        -\zeta^{-2}\grad\varphi(X) & \zeta^{-2}\grad\varphi(X)\grad\varphi(X)^\top + \zeta^{-1} \grad^2\varphi(X) - \grad^2\log\det(X)
    \end{bmatrix} \nonumber \\
    &= \begin{bmatrix} \zeta^{-1} \\ -\zeta^{-1}\grad\varphi(X) \end{bmatrix}\begin{bmatrix} \zeta^{-1} \\ -\zeta^{-1}\grad\varphi(X) \end{bmatrix}^\top + \begin{bmatrix}
        0 & 0 \\
        0 & \zeta^{-1} \grad^2\varphi(X) - \grad^2\log\det(X)
    \end{bmatrix}.
\end{align}
It is fairly straightforward to show, using block elimination, that the solution to the linear system $\grad^2 F(t, X)[s, H] = (u, V)$ is given by
\begin{subequations}\label{eqn:barrier-inv}
    \begin{align}
        H &= (\zeta^{-1} \grad^2\varphi(X) - \grad^2\log\det(X))^{-1}(V + u\grad\varphi(X)) \\
        s &= \zeta^2u + \inp{\grad\varphi(X)}{H},
    \end{align}
\end{subequations}
i.e., the main effort is in solving a linear system with the matrix $\zeta^{-1} \grad^2\varphi(X) - \grad^2\log\det(X)$, which is guaranteed to be positive definite (as $X$ is positive definite and $\varphi$ is convex).

\section{\diff{Newton system solver}}\label{appdx:pdipm}

\diff{In this section, we outline the approach we take to solving the Newton systems arising in each step of a primal-dual interior-point method. Note that we follow the standard approach of reducing the system to normal equation form, then focusing on solving linear systems with the Schur complement matrix e.g.,~\cite{wright1997primal,skajaa2015homogeneous}). Unless otherwise specified, in our experiments, we replaced Hypatia's default approach to solving Newton systems (which uses the method from~\cite[Section 10.3]{vandenberghe2010cvxopt}) with the approach outlined in this section.}

Consider the standard form primal and dual conic programs\\
\begin{minipage}{0.5\textwidth}
\begin{subequations}\label{eqn:primal}
    \begin{align}
        \mathmakebox[\widthof{$\subjto$}][c]{\minimize_{x\in\mathbb{R}^n}} \quad & c^\top x \\
        \subjto \quad & Ax = b \label{eqn:primal-b} \\
                \quad & x \in \mathcal{K},
    \end{align}
\end{subequations}
\end{minipage}%
\begin{minipage}{0.5\textwidth}
\begin{subequations}\label{eqn:dual}
    \begin{align}
        \maximize_{y\in\mathbb{R}^p, z\in\mathbb{R}^n} \quad & -b^\top y \\
        \mathmakebox[\widthof{$\maximize_{y\in\mathbb{R}^p, z\in\mathbb{R}^n}$}][c]{\subjto} & z - A^\top y = c \\
                \quad & z \in \mathcal{K}_*,
  \end{align}
\end{subequations}
\end{minipage}\bigskip

\noindent where $c\in\mathbb{R}^n$, $A\in\mathbb{R}^{p\times n}$, $b\in\mathbb{R}^p$, $\mathcal{K}\subseteq\mathbb{R}^n$ is a proper convex cone, and $\mathcal{K}_*\subseteq\mathbb{R}^n$ is the dual cone of $\mathcal{K}$. Usually, we will consider $\mathcal{K}$ to be the Cartesian product of multiple proper cones, i.e., $\mathcal{K}=\mathcal{K}_1\times\ldots\times\mathcal{K}_m$, where $\mathcal{K}_i$ is a proper cone for all $i=1,\ldots,m$.

\diff{Let} $F:\interior\mathcal{K}\rightarrow\mathbb{R}$ be a $\nu$-logarithmically homogeneous self-concordant barrier for the cone $\mathcal{K}$. The central path \diff{$\{(x(\mu), y(\mu), z(\mu)) : \mu>0\}$, parameterized by a barrier parameter $\mu>0$,} can be characterized by the following KKT equations \diff{(see, e.g.,~\cite{skajaa2015homogeneous})}
\begin{subequations}\label{eqn:central-path}
    \begin{align}
        z(\mu) - A^\top y(\mu) &= c\\
        Ax(\mu)     &= b \\
        \mu \grad F(x(\mu))  + z(\mu) &= 0.
    \end{align}
\end{subequations}
Primal-dual interior-point methods \diff{follow the central path to $x\downarrow0$} by taking appropriate linearizations of the central path equations. These steps $(\Delta x, \Delta y, \Delta z)$ satisfy linear equations (also known as Newton systems) of the form
\begin{align}
    \begin{bmatrix} 
        0  & -A^\top & \mathbb{I} \\ 
        -A & 0      & 0      \\
        \mu \grad^2F(x) & 0      & \mathbb{I}
    \end{bmatrix} \begin{bmatrix} \Delta x \\ \Delta y \\ \Delta z \end{bmatrix} = \begin{bmatrix} r_x \\ r_y \\ r_z \end{bmatrix},
\end{align}
for some right-hand side $(r_x, r_y, r_z)$ which are dependent on how we linearize the central path equations, e.g., if we want to solve for a direction tangential to the central path, or to solve the central path equations as a system of nonlinear equations. If we perform block elimination on this system of equations, we obtain
\begin{subequations}\label{eqn:pdipm-block-elim}
    \begin{align}
        A\grad^2F(x)^{-1}A^\top\Delta y &=  A\grad^2F(x)^{-1}(r_z - r_x) - \mu r_y  \\
        \Delta x &= \frac{1}{\mu}\grad^2F(x)^{-1} (r_z - r_x - A^\top\Delta y) \\
        \Delta z &= A^\top \Delta y + r_x.
    \end{align}
\end{subequations}
Taking this approach, the main cost in solving for a stepping direction is in forming and Cholesky factoring the Schur complement matrix $A\grad^2F(x)^{-1}A^\top$. 

\diff{Note that in our actual experiments, we do not solve this exact version of the Newton systems, but a more complicated variant which uses a homogeneous self-dual embedding. Doing this has some numerical and practical advantages (see, e.g.,~\cite{skajaa2015homogeneous,coey2023performance,vandenberghe2010cvxopt}). However, the main ideas used in solving the Newton systems remain largely the same as those presented in this appendix.}

\section{Additional details on the discrete-phase-randomized protocol}\label{appdx:dprBB84}
Here, we provide some additional details about the structure of the linear maps for the dprBB84 quantum key rate protocol. Recall that for this protocol, we are interested in the linear map
\begin{equation*}
    \mathcal{G}(X) = K_1 X K_1^\dag + K_2 X K_2^\dag,
\end{equation*}
where $K_1$ and $K_2$ are given by~\eqref{eqn:dprbb84-kraus}. A first observation is that $K_1^\dag K_2 = K_2^\dag K_1 = 0$, and therefore $\inp{K_1XK_1^\dag}{K_2XK_2^\dag}=0$ for all $X\in\mathbb{H}^{12c}$. This implies that $\mathcal{G}(X)$ has a block diagonal structure. More concretely, we can show that
\begin{equation}\label{eqn:dpr-g}
    \mathcal{G}(X) = P_g \begin{bmatrix}
        \mathcal{G}_1(X) & & \\ & \mathcal{G}_2(X) & \\ & & 0
    \end{bmatrix} P_g^\top,
\end{equation}
where $P_g$ is a suitable permutation matrix, and $\mathcal{G}_1(X)\coloneqq \hat{K}_1X\hat{K}_1^\dag$ and $\mathcal{G}_2(X)\coloneqq \hat{K}_2X\hat{K}_2^\dag$ where
\begin{subequations}
    \begin{align}
        \hat{K}_1 = \sqrt{p_z}& \mleft(\bigoplus_{i=1}^c \begin{bmatrix}
            1 & 0 & 0 & 0 \\ 0 & 1 & 0 & 0
        \end{bmatrix} \mright) \otimes \begin{bmatrix}
            1 & 0 & 0 \\ 0 & 1 & 0
        \end{bmatrix} \\
        \hat{K}_2 = \sqrt{1 - p_z} &\mleft(\bigoplus_{i=1}^c \begin{bmatrix}
            0 & 0 & 1 & 0 \\ 0 & 0 & 0 & 1
        \end{bmatrix} \mright) \otimes \begin{bmatrix}
            1 & 0 & 0 \\ 0 & 1 & 0
        \end{bmatrix}.
    \end{align}
\end{subequations}
Next, we also recognize that the pinching map, the act of zeroing off-diagonal blocks, can only ever split a block diagonal block into multiple smaller block diagonal blocks. Therefore, we expect $Z_jK_iXK_i^\dag Z_j^\dag$ for $i,j=1,2$ to represent four distinct block diagonal blocks of the image of $\mathcal{Z}\circ\mathcal{G}$, where $Z_j=e_je_j^\top\otimes\mathbb{I}$ for $j=1,2$. More concretely, we can show that
\begin{equation}\label{eqn:dpr-z}
    \mathcal{Z}(\mathcal{G}(X)) = P_z \begin{bmatrix}
        \mathcal{Z}_1(\mathcal{G}_1(X)) & & & & \\ & \mathcal{Z}_2(\mathcal{G}_1(X)) & & & \\ & & \mathcal{Z}_1(\mathcal{G}_2(X)) & & \\ & & & \mathcal{Z}_2(\mathcal{G}_2(X)) & \\ & & & & 0
    \end{bmatrix} P_z^\top,
\end{equation}
where $P_z$ is a suitable permutation matrix, and $\mathcal{Z}_1(X)\coloneqq \hat{Z}_1X\hat{Z}_1^\dag$ and $\mathcal{Z}_2(X)\coloneqq \hat{Z}_2X\hat{Z}_2^\dag$ where
\begin{subequations}
    \begin{align}
        \hat{Z}_1 &= \mleft(\bigoplus_{i=1}^c \begin{bmatrix}
            1 & 0
        \end{bmatrix} \mright) \otimes \mathbb{I} \\
        \hat{Z}_2 &= \mleft(\bigoplus_{i=1}^c \begin{bmatrix}
            0 & 1
        \end{bmatrix} \mright) \otimes \mathbb{I}.
    \end{align}
\end{subequations}
Note that the linear maps $\mathcal{G}_1$, $\mathcal{G}_2$, $\mathcal{Z}_1$, and $\mathcal{Z}_2$ all have the effect of extracting a principal submatrix from a matrix, up to a scaling factor. Therefore, computationally, these maps are easy to apply as they only require indexing and scalar multiplication operations.

\section{Additional details on experimental setup}\label{appdx:model}
Here, we provide additional details about how we model the quantum relative entropy \diff{optimization problems} to parse into Hypatia for the numerical experiments we present in Section~\ref{sec:exp}. Our proposed methods all model the problem in the standard form~\eqref{eqn:primal} and solve the Newton equations using the block elimination method described in Appendix~\ref{appdx:pdipm}. On the other hand, for many of the benchmark methods we compare our proposed approaches to, we found it computationally advantageous to give problems to Hypatia in the more general form 
\begin{equation}
    \minimize_{x\in\mathbb{R}^n} \quad \inp{c}{x} \quad \subjto \quad Ax = b, \quad h-Gx \in\mathcal{K},
\end{equation}
where $c\in\mathbb{R}^n$, $A\in\mathbb{R}^{p\times n}$, $b\in\mathbb{R}^p$, $G\in\mathbb{R}^{q\times n}$, $h\in\mathbb{R}^q$, and $\mathcal{K}\subset\mathbb{R}^q$ is a proper convex cone, and to solve problems using Hypatia's default Newton system solver, which is based on the method from~\cite[Section 10.3]{vandenberghe2010cvxopt}.

\subsection{Quantum key distribution}
For the QRE method, we model the problem as the standard quantum relative entropy \diff{optimization problem}
\begin{equation}\label{eqn:dprbb84-qre}
    \begin{aligned}
        \minimize_{t,X} \quad & t \\
        \subjto \quad & \mathcal{A}(X) = b \\
        & (t, V^\dag\mathcal{G}(X)V, V^\dag\mathcal{G}(\mathcal{Z}(X))V) \in \mathcal{K}_{\textnormal{qre}} \\
        & X \succeq 0,
    \end{aligned}
\end{equation}
where $V$ is a suitable isometry performing facial reduction based on~\cite[Lemma 6.2]{karimi2023efficient}. For the DPR and QKD methods, we model the quantum key distribution problem in the same way using a suitably parameterized cone from Corollary~\ref{cor:main}, i.e.,
\begin{equation}\label{eqn:dprbb84-dpr}
    \begin{aligned}
        \minimize_{t,X} \quad & t \\
        \subjto \quad & \mathcal{A}(X) = b \\
        & (t, X) \in \mathcal{K}_{\textnormal{qre}}^{\mathcal{G}, \mathcal{Z}\circ  \mathcal{G}}.
    \end{aligned}
\end{equation}

\subsection{Quantum rate-distortion}
For the QRE method, we model the quantum rate-distortion problem as
\begin{equation}\label{eqn:qrd-qre}
    \begin{aligned}
        \minimize_{t,X} \quad & t \\
        \subjto \quad & \tr_2^{n,m}(X) = W \\
        & (t, X, \mathbb{I}\otimes\tr_1^{n,m}(X)) \in \mathcal{K}_{\textnormal{qre}} \\
        & \inp{\Delta}{X} \leq D.
    \end{aligned}
\end{equation}
Alternatively, as we do for QCE, we can directly model the problem using the quantum conditional entropy cone from Corollary~\ref{cor:qce-cone} as
\begin{equation}\label{eqn:qrd-qce}
    \begin{aligned}
        \minimize_{t,X} \quad & t \\
        \subjto \quad & \tr_2^{n,m}(X) = W \\
        & (t, X) \in \mathcal{K}_{\textnormal{qce}} \\
        & \inp{\Delta}{X} \leq D.
    \end{aligned}
\end{equation}

\subsubsection{Entanglement fidelity distortion}
There are a few ways we can take into account the fixed point subspace~\eqref{eqn:qrd-fixed-point-subspace}. As we do for QCE*, the simplest way is to work directly within the image of the subspace as follows
\begin{equation}\label{eqn:efqrd-qce}
    \begin{aligned}
        \minimize_{t,y,Z} \quad & t \\
        \subjto \quad & \tr_2^{n,n}(\mathcal{G}(y,Z)) = W \\
        & (t, \mathcal{G}(y,Z)) \in \mathcal{K}_{\textnormal{qce}} \\
        & \inp{\Delta}{\mathcal{G}(y,Z)} \leq D,
    \end{aligned}
\end{equation}
where $\mathcal{G}$ is given by~\eqref{eqn:qrd-lin-map}. Alternatively, as we do for QRE*, we can use~\eqref{eqn:qrd-ef-double-re} to model the problem using classical and quantum relative entropies
\begin{equation}\label{eqn:efqrd-qre}
    \begin{aligned}
        \minimize_{t_1,t_2,y,Z} \quad & t_1 + t_2 \\
        \subjto \quad & \tr_2^{n,n}(\mathcal{G}(y,Z)) = W \\
        & (t_1, y, \hat{\mathcal{G}}_1(y, Z)) \in \mathcal{K}_{\textnormal{qre}} \\
        & (t_2, Z, \hat{\mathcal{G}}_2(y, Z)) \in \mathcal{K}_{\textnormal{cre}} \\
        & \inp{\Delta}{\mathcal{G}(y,Z)} \leq D.
    \end{aligned}
\end{equation}
Finally, we can directly use the tailored cone from Corollary~\ref{cor:cone-qrd} to model the problem, which we use for the QRD method
\begin{equation}\label{eqn:efqrd-qrd}
    \begin{aligned}
        \minimize_{t,y,Z} \quad & t \\
        \subjto \quad & \tr_2^{n,n}(\mathcal{G}(y,Z)) = W \\
        & (t, y, Z) \in \mathcal{K}_{\textnormal{qrd}} \\
        & \inp{\Delta}{\mathcal{G}(y,Z)} \leq D.
    \end{aligned}
\end{equation}

\subsection{Quantum channel capacity}

\subsubsection{Entanglement-assisted capacity}
To compute the entanglement-assisted channel capacity, for the na\"ive QRE method we model the problem using the quantum relative entropy and quantum entropy cones as
\begin{equation}\label{eqn:eacc-qre}
    \begin{aligned}
        \minimize_{t_1,t_2,X} \quad & t_1 + t_2 \\
        \subjto \quad & \tr[X] = 1 \\
        & (t_1, VXV^\dag, \mathbb{I}\otimes\tr_1^{m,p}(VXV^\dag)) \in \mathcal{K}_{\textnormal{qre}} \\
        & (t_2, \tr_2^{m,p}(VXV^\dag), \tr[X]) \in \mathcal{K}_{\textnormal{qe}}.
    \end{aligned}
\end{equation}
For the QCE method, we model the problem using the quantum conditional entropy cone as
\begin{equation}\label{eqn:eacc-qce}
    \begin{aligned}
        \minimize_{t_1,t_2,X} \quad & t_1 + t_2 \\
        \subjto \quad & \tr[X] = 1 \\
        & (t_1, VXV^\dag) \in \mathcal{K}_{\textnormal{qce}} \\
        & (t_2, \tr_2^{m,p}(VXV^\dag), \tr[X]) \in \mathcal{K}_{\textnormal{qe}}.
    \end{aligned}
\end{equation}
Alternatively, for the QMI method we directly optimize over the quantum mutual information cone from Corollary~\ref{cor:qmi-cone} as follows
\begin{equation}\label{eqn:eacc-qmi}
    \begin{aligned}
        \minimize_{t,X} \quad & t \\
        \subjto \quad & \tr[X] = 1 \\
        & (t, X) \in \mathcal{K}_{\textnormal{qmi}}.
    \end{aligned}
\end{equation}

\subsubsection{Quantum-quantum capacity of degradable channels}
Similarly, for the quantum-quantum channel capacity, in the QRE method, we model the problem using the full quantum relative entropy cone as
\begin{equation}\label{eqn:qqcc-qre}
    \begin{aligned}
        \minimize_{t,X} \quad & t \\
        \subjto \quad & \tr[X] = 1 \\
        & (t, W\mathcal{N}(X)W^\dag, \mathbb{I}\otimes\tr_1^{p,q}(W\mathcal{N}(X)W^\dag)) \in \mathcal{K}_{\textnormal{qre}} \\
        & X \succeq 0.
    \end{aligned}
\end{equation}
Alternatively, in the QCE method, we model the problem using the quantum conditional entropy cone as
\begin{equation}\label{eqn:qqcc-qce}
    \begin{aligned}
        \minimize_{t,X} \quad & t \\
        \subjto \quad & \tr[X] = 1 \\
        & (t, W\mathcal{N}(X)W^\dag) \in \mathcal{K}_{\textnormal{qce}} \\
        & X \succeq 0.
    \end{aligned}
\end{equation}
In the QCI method, we instead directly minimize over the quantum coherent information cone from Corollary~\ref{cor:qci-cone} as
\begin{equation}\label{eqn:qqcc-qci}
    \begin{aligned}
        \minimize_{t,X} \quad & t \\
        \subjto \quad & \tr[X] = 1 \\
        & (t, X) \in \mathcal{K}_{\textnormal{qci}}.
    \end{aligned}
\end{equation}

\subsection{Ground state energy of Hamiltonians}
To compute compute lower bounds for the ground state energy of Hamiltonians, we can either model the problem using the quantum relative entropy cone as in the QRE method,
\begin{equation}\label{eqn:med-qre}
    \begin{aligned}
        \minimize_{X} \quad & \inp{H}{X} \\
        \subjto \quad & \tr_1^{2,2^{l-1}}(X) = \tr_2^{2^{l-1},2}(X) \\
        & \tr[X] = 1 \\
        & (0, X, \mathbb{I}\otimes\tr_1^{2,2^{l-1}}(X)) \in \mathcal{K}_{\textnormal{qre}},
    \end{aligned}
\end{equation}
or, as we do for the QCE method, we can directly optimize over the quantum conditional entropy function from Corollary~\ref{cor:qce-cone} 
\begin{equation}\label{eqn:med-qce}
    \begin{aligned}
        \minimize_{t,X} \quad & \inp{H}{X} \\
        \subjto \quad & \tr_1^{2,2^{l-1}}(X) = \tr_2^{2^{l-1},2}(X) \\
        & \tr[X] = 1 \\
        & t = 0 \\
        & (t, X) \in \mathcal{K}_{\textnormal{qce}}.
    \end{aligned}
\end{equation}
}

\bibliographystyle{IEEEtran}
\bibliography{refs}

\end{document}